%% file: AA_doc.tex
\documentclass[journal,twoside,web]{ieeetran}

\usepackage[utf8]{inputenc} 
\usepackage[T1]{fontenc}    
\usepackage{url}            
\usepackage{booktabs}       
\usepackage{amsfonts}       
\usepackage{nicefrac}       
\usepackage{microtype}      
\usepackage{amsthm}
\usepackage{cite}

\usepackage{amsmath}
\usepackage{amssymb,mathrsfs}
\usepackage{nccmath}
\usepackage{mathtools}
\usepackage{xcolor}
\usepackage{graphicx} 

\usepackage{wrapfig}
\usepackage{enumitem}
\usepackage{subfig}
\usepackage{algcompatible,algorithm}

\usepackage{cleveref}

\newtheorem{theorem}{Theorem}
\newtheorem{assumption}{Assumption}
\newtheorem{remark}{Remark}

\newcommand{\tr}{\mathrm{tr}}
\newcommand{\Pa}{\mathrm{P_1}}
\newcommand{\Pb}{\mathrm{P_2}}

\newcommand{\dx}{\mathrm{d}x}
\newcommand{\dt}{\mathrm{d}t}
\newcommand{\ds}{\mathrm{d}s}
\newcommand{\dW}{\mathrm{d}W}
\newcommand{\de}{\mathrm{d}e}
\newcommand{\dz}{\mathrm{d}z}

\renewcommand{\baselinestretch}{1}

\makeatletter

\input{arxivmacros.tex}

\title{\huge Communication and Control Co-design \\ in Non-cooperative Games}

\author{Shubham~Aggarwal, Tamer~Ba{\c s}ar, and Dipankar~Maity 
\thanks{Research of the authors was supported by the ARL grant ARL DCIST CRA W911NF-17-2-0181. Research of the first two authors was supported by the AFOSR grant 476795-239015-191100. 
}
\thanks{
Shubham Aggarwal is with the Department of Mechanical Science \& Engineering and the Coordinated Science Laboratory at the University of Illinois Urbana-Champaign (UIUC), Urbana, IL, USA-61801; Tamer Ba{\c s}ar is with the Department of Electrical \& Computer Engineering and the Coordinated Science Laboratory at UIUC;
Dipankar Maity is with Department of Electrical and Computer Engineering at the University of North Carolina at Charlotte, Charlotte, NC, USA-28223.
(Emails:\texttt{\{sa57, basar1\}@illinois.edu, dmaity@charlotte.edu)}
}
}
\date{}

\begin{document}

\maketitle
\thispagestyle{empty}
\begin{abstract}
In this article, we revisit a communication-control co-design problem for a class of two-player stochastic differential games on an infinite horizon. Each `player' represents two active decision makers, namely a \textit{scheduler} and a \textit{remote controller}, which cooperate to optimize over a global objective while competing with the other player. 
Each player's scheduler can only \textit{intermittently} relay state information to its respective controller due to associated cost/constraint to communication.
The scheduler's policy determines the \textit{information structure} at the controller, thereby affecting the quality of the control inputs.
Consequently, it leads to the classical \textit{communication-control trade-off} problem.
A high communication frequency improves the control performance of the player on account of a higher communication cost, and vice versa.
Under suitable information structures of the players, we first compute the Nash controller policies for both players in terms of the conditional estimate of the state. 
Consequently, we reformulate the problem of computing Nash scheduler policies (within a class of parametrized randomized policies) into solving for the steady-state solution of a generalized Sylvester equation. Since the above-mentioned reformulation involves infinite sum of powers of the policy parameters, we provide a projected gradient descent-based algorithm to numerically compute a Nash equilibrium using a truncated polynomial approximation. Finally, we demonstrate the performance of the Nash control and scheduler policies using extensive numerical simulations.
\end{abstract}


\section{Introduction}
The topic of ``games among teams'' has been a widely addressed one due to its ubiquitous applicability in diverse systems. Each team within this setup is comprised of multiple decision makers (DMs), which collectively aim to satisfy a joint cooperative objective while playing in cooperation or competition with other teams. A state-of-the-art application of the cooperative setup pertains to cognitive radio networks \cite{gao2024attention}, where players collaboratively perform joint spectrum sensing followed by decentralized spectrum access. Another application within the cooperative setup is that of decentralized federated learning \cite{meng2023gnn}  where players exchange weights of a neural network over multiple rounds for collaborative training without central coordination. Within the non-cooperative setup, a classic example is that of the perimeter monitoring problem \cite{shishika2020review} in between two teams, where the objective of a defender team is to defend a target of interest against an attacker team. Another application of such a setup is that in network security games \cite{alpcan2010network} where the security firewall (the defender) is responsible for neutralizing any attacks from a malicious attacker. Other important applications include oligopoly markets, smart and connected cities, field robotics, and the like \cite{weintraub2008markov,hatanaka2015passivity,chi2021game}.

In this paper, we consider a two-player differential game problem, where each player\footnote{In the sequel, we use the terms player and team interchangeably.} is essentially a wireless control system representing two decision makers, namely, a scheduler and a controller (see Fig. \ref{Fig:system_model}).
The scheduler accesses the state of the game and communicates it to its remote controller counterpart to \textit{exert a control action}. 
Each communication as well as a control action incur associated costs. 
Thus, each player is faced with the two objectives of deciding when to control and when to communicate. Increased communication between the scheduler and the controller leads to a better estimate of the state at the controller, which improves control performance, but at the expense of increased communication cost. On the other hand, a lower communication frequency hurts control performance whilst potentially saving on the communication cost. 
This gives rise to the above-mentioned communication-control tradeoff for each player. 

The distributed actions of both players in controlling the state of the game couples the two players' optimization problems, and hence, Nash equilibrium is an appropriate solution concept within this two-player setting. Our objective in this work is, thus, to find Nash controller and Nash scheduler policies for the players. Under suitable information structures, we first compute closed-form expressions for the Nash controller policies in terms of the conditional estimates of the state of the game at each player. Subsequently, we re-express the problem of obtaining the Nash scheduler policies (using a class of parametrized randomized policies) into solving for the steady-state solution of a generalized Sylvester-type equation. Since the same leads to a Neumann series involving a sum of infinite powers of the policy parameters, we first provide an approximate solution to the same using truncated polynomial expressions. Consequently, we provide a projected gradient-descent type algorithm to compute a Nash equilibrium solution for the scheduler of each player.

\subsection{Related Work}
\begin{figure}[t]
	\centering
	\includegraphics[width=0.75 \columnwidth ]{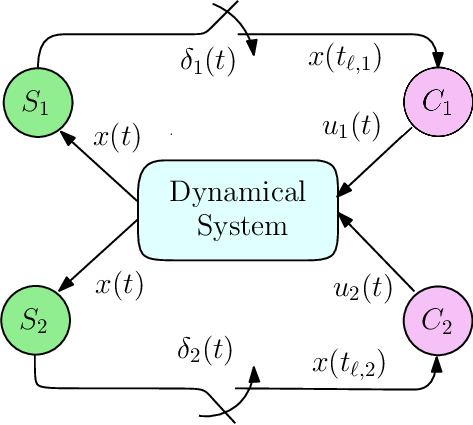}
	\caption{\small{Schematic of a two-player wireless system exerting control on the state of a dynamical system.
    	}}
	\label{Fig:system_model}
\end{figure}
\subsubsection{Networked Control systems}
Networked control systems (NCSs) form a special instance of the games-among-teams setting constituting only one player, equipped with a \textit{controller} and a \textit{scheduler} \cite{zhang2019networked}. 
Existing literature on NCSs focus on two perspectives, namely those of event-triggered control and optimal control. Within the event-triggered framework \cite{ demirel2016trade,peng2018survey,yi2018dynamic,heemels2021event}, the key objective is that of system stabilization, and is based on pre-defined scheduler policies followed by a Lyapunov-type analysis. Extensions to cooperative multi-agent NCSs include \cite{li2020consensus,xu2023event}, among many others (see also, references \cite{wangsurvey,masonmulti,liu2024event}). Evidently, such designs lack any optimality guarantees.

On the other hand, optimal control of the NCSs has also been widely addressed in literature \cite{imer2005optimal,imer2010optimal,lipsa2011remote,imer2006optimal,imer2006measure,molin2012optimality,molin2016optimality,maity2020minimal} to determine \textit{optimal} controller and scheduler policies.
Specifically, papers \cite{imer2005optimal} and \cite{imer2010optimal} have obtained the optimal sensing policies within the class of threshold-type (symmetric) policies under a hard constraint on the number of communications. In a subsequent work \cite{lipsa2011remote}, the former policy space was proven to be optimal within a general class of policies, for scalar systems, and with a penalty rather than a hard constraint on the communications.
Later, \cite{imer2006optimal} derived optimal policies under resource constraints on the interaction between the controller and the actuator in a plant, and \cite{imer2006measure} introduced a tradeoff between two options for a controller, which are transmission of control signals to the actuator and receiving measurements, where optimality is again based on threshold-type policies. The works 
\cite{molin2012optimality,molin2016optimality} provided jointly optimal policies for control and sensing for linear multivariate systems based on the estimation error while \cite{maity2020minimal} derived suboptimal sensor and control policies based on the statistics of the estimation error.
The latest works in this direction for single player systems are \cite{soleymani2021value,soleymani2022value}, where the authors show that for a multi-dimensional Gauss-Markov process, the optimal estimator is linear and is unaffected by the \textit{no-communication events} (as was derived for \textit{scalar} systems earlier in \cite{imer2010optimal} with its optimality proven in \cite{lipsa2011remote}), leading to a globally optimal policy design for a single team problem. The scheduling policy expression, is however, provided in terms of the conditional value function of the state, which cannot be computed exactly.

\subsubsection{Networked games} 
Earlier works within this domain addressed the design of controller policies for a pre-specified set of scheduling instants (also referred to as sampled-data information structures) for both deterministic and stochastic systems, and for both zero-sum and general sum games \cite{bacsar1977informationally,bacsar1991game,bacsar1991optimum,bacsar1995minimax,bacsar2005existence}.
Following the lines of work on event-triggered control, the work \cite{yuan2017event} considered a similar problem for quadratic games, with \textit{pre-defined} sensing policies for the players.
The first work discussing the concurrent design of controller-scheduler policies for zero-sum linear quadratic (LQ) games appears to be \cite{maity2016strategies}. Under the setting, the controllers of both players have only sampled (and synchronous) state access, and the authors derive a guaranteed performance bound for the underlying objective function. Equilibrium scheduling policy design was, however, left open. 
These results on controller policies were improved in \cite{maity2016optimal} to derive saddle-point control strategies for the case where information could leak between players. This caused the underlying scheduling policy design problem to become a joint optimization problem between players (instead of a game) and resulted in a \textit{cooperative} scheduling strategy.

Results on asynchronous scheduling problems, where both players could possibly have different scheduling instances can be found in \cite{maity2017linear}. The authors derive equilibrium controller strategies for a 2--player game, but the equilibrium scheduling policy computation turned out to be intractable. This framework was also adopted in \cite{huang2021defending} for asset defending applications where the players have a decoupled dynamics, thereby turning out to be a special case of the setup of \cite{maity2017linear}.
Due to this decoupling, the scheduling strategy computation becomes an optimization problem similar to what was proposed in \cite{maity2016strategies}. 
One of the latest works \cite{aggarwal2024linear} considers a linear quadratic zero-sum differential game where one player has continuous while the other player has intermittent state access. Within the same, the authors derive age-of-information based scheduling policy for the player at a disadvantage, which serves as its security strategy, under the assumption that the other player is oblivious to the sensing limitation of the disadvantaged player. Computing equilibrium strategies in such a setting is, however, still an open problem. Additionally, another work worth mentioning is \cite{aggarwal2024semantic} where the authors consider a `large' player system, and compute \textit{approximate} Nash equilibrium control and scheduling policies using the framework of mean-field game theory.

The above discussion leads one to the conclusion that designing joint equilibrium control-scheduler policy pairs is still a largely open problem, from both analytical and computational points of view. 
Motivated by these challenges, this paper explores (Nash) equilibrium computation in a two-player nonzero-sum game, where each player represents a scheduler-controller pair. In this framework, we allow for asynchronous scheduling for both players with no information leakage when scheduling decisions are made. Within this general setup, our key contributions are as follows:
\begin{enumerate}
    \item \textbf{Existence of Equilibrium Controller Policies:} We derive sufficient conditions (in Theorem \ref{thm:Nash_control}) that guarantee the existence of equilibrium controller policies for both players, where each policy is characterized as a function of the conditional estimate of the game state. The major challenge here lies in rigorously establishing the  structure of the control policy and of the state estimate.

    \item  \textbf{Scheduling Policy Reformulation:} To tackle the complexity of scheduling, we introduce a class of randomized open-loop scheduling policies, reformulate the policy design problem using a generalized Sylvester equation (in Proposition \ref{Prop:cov_dynamics}), and subsequently discuss solvability conditions (in Theorem \ref{thm:approximation}) to characterize the steady-state behavior of scheduling policies (which turn out to be an infinite degree polynomial function of the decision variables). Thus, to ensure tractability in absence of a closed-form policy expression, we propose a gradient-descent-based algorithm to determine scheduling instances by applying an approximate polynomial truncation of the aforementioned steady-state solution.
\end{enumerate}




\textbf{Organization: } The rest of the paper is organized as follows. We formulate the two-player differential game problem in Section~\ref{sec:Problem_form} and provide some preliminary results in Section~\ref{sec:prelim}. The equilibrium estimator-controller policies for both players are obtained in Section~\ref{sec:Controller}. The scheduler optimization problem and the corresponding equilibrium policy are detailed in Section~\ref{sec:Sensor}.
We provide supporting numerical simulations in Section \ref{sec:numSims} and conclude the paper with its major highlights in Section \ref{sec:conc_disc}, followed by extended proofs in supporting Appendices.

\textbf{Notations:} For a given time $t$, $t^-$ denotes the time right before $t$. For symmetric matrices $X$ and $Y$, the notation $X \succeq Y$ (resp. $X \succ Y$) says that $X-Y$ is positive semi-definite (resp. positive definite).  
$\N(\mu, \Sigma)$ denotes Gaussian distribution with mean $\mu$ and covariance $\Sigma$.
The notation ${\tt card}(S)$ denotes the cardinality of a countable set $S$. $\mathbb N:= \{1, 2, \cdots \}$ denotes the set of natural numbers. We denote a block diagonal (square) matrix $X \in \R^{(n_1+n_2)\times (n_1 + n_2)}$ with diagonal entries $X_1 \in \R^{n_1 \times n_1}$ and $X_2 \in \R^{n_2 \times n_2}$ as $X = {\tt Diag}[X_1, X_2]$.

\section{Problem Formulation}\label{sec:Problem_form}
We consider a two-player stochastic linear-quadratic (LQ) differential game where the players (referred to as $\Pa$ and $\Pb$) can exert control (denoted $u_1$ and $u_2$) over the state of the game $x$, which evolves according to the following linear stochastic differential equation:
\begin{align} \label{eq:dyn}
\begin{split}
    \dx(t) & = [Ax(t) + B_1u_1(t) + B_2u_2(t)] \dt + G \dW(t), \\
    x(0) & \sim \N(0,\Sigma_0)
\end{split}
\end{align}
with $x \in \Re^{n_x}$, $u_i \in \Re^{n_{u_i}}$, and $\{W(t) \in \Re^{n_w}\}_{t \geq 0}$ is an $n_w$--dimensional standard Brownian motion which is independent of the  initial state $x(0)$. 
The control objective over an infinite horizon is given by:
\begin{align} \label{eq:pre_control_cost}
    J = \limsup_{T \rightarrow \infty} \frac{1}{T}\E \Big[ \int_0^Tc(x(t),u_1(t),u_2(t)) \dt\Big],
\end{align}
where the instantaneous objective function is a quadratic one:
\begin{align} \label{eq:control_cost}
    c(x,u_1,u_2) := \|x\|_Q^2 + \|u_1\|^2_{R_1} - \|u_2\|^2_{R_2},
\end{align}
where $Q \succeq 0, R_i \succ 0$, for $i=1,2$. Player $\Pa$ aims to minimize the objective function through its control policy $u_1 (\cdot)$, whereas $\Pb$ aims to maximize the objective function by selecting $u_2(\cdot)$. 

\subsection{Communication-constrained LQ Games}
In contrast to standard perfect-state LQ games, where the players have continuous access to the states $x(t)$ \cite{bacsar1977informationally}, or some partial measurements $y(t)$ \cite{bacsar2014stochastic}, our formulation here focuses on the scenario where each player actively schedules communication over the wireless medium to its remote controller counterpart.
Each such scheduling is costly, and therefore, must be performed judiciously. 
Classic examples of such a setup include pursuit-evasion games \cite{bacsar1998dynamic} and perimeter monitoring games \cite{shishika2020review}, among others.

In this continuous time framework, the players are allowed to schedule the communication in a discrete time fashion. 
That is, a player can have only access to $\{x(t_k)\}_{k \in \mathbb{N}}$, where the discrete scheduling instances $t_k$'s are decided by the player. 
Let $\lambda_{ii} > 0$ denote the scheduling cost per time instant for each Player $i$. 
In addition to that, let $\lambda_{ij} \in \R$ be the cost/payoff for Player $i$ every time Player $j$ communicates the measured state to its controller.
We assume that the scheduler communicates the measured state instantly (i.e., without delay) over a reliable communication channel (i.e., no packet dropout). 
Let us define the set of (possibly random) scheduling instants up to time $t$ for player $i$ as
\begin{align}
    {\tt T}_i(t) := \{t_{\ell,i} \mid \ell = 1, \ldots, n_i(t)\}
\end{align}
where $n_i(t) \in \mathbb{N}$ is the total number of communications up to time $t$. 
Furthermore, we have that $t_{\ell,i}  < t_{\ell',i}$ for $\ell < \ell'$, almost surely. 
Henceforth, we will let ${\tt T}_i := \lim_{T\to \infty} {\tt T}_i(T)$ for $i=1,2,$ and $\tt T$ to denote ${\tt T}_1 \cup {\tt T}_2$.

The incurred communication cost of each player up to time $t$ is given as:
\begin{align}\label{eq:sensing_cost}
    m_i(t) := \lambda_{ii} n_i(t) + \lambda_{ij} n_{j}(t).
\end{align}
As mentioned earlier, the scheduling instances $t_{\ell,i}$ are optimization variables for Player $i$. 
Let $\delta_i(t) \in \{0,1\}$ denote the scheduling decision variable for each player $i$ at time $t$.
That is, $t_{\ell, i} \in {\tt T_i}$ if, and only if, $\delta(t_{\ell, i}) = 1$ for all $\ell$ and player $i$. 

\begin{assumption}
    At every time $t$, each player knows ${\tt T}(t) = {\tt T}_1(t) \cup {\tt T}_2(t)$.
\end{assumption}
This assumption states that each player also knows the communication instances of their opponents. 
However, they do not know the communicated measurement value that the opponent's controller had received.


\subsection{Joint Control-Communication Decision-making}
Combining both the communication scheduling and control objectives, the objective functions of the players become as follows, where $\Pa$ minimizes $J_1$ and $\Pb$ maximizes $J_2$:
\begin{align}
    & J_1  = \limsup_{T \rightarrow \infty} \frac{1}{T}\E \Big[\! \int_0^T \!\!\!\! c(x(t),u_1(t),u_2(t)) \dt + \alpha_1 m_1(T) \Big], \label{eq:cost_P1} \\
    & J_2  = \limsup_{T \rightarrow \infty} \frac{1}{T}\E \Big[ \int_0^T \!\!\!\! c(x(t),u_1(t),u_2(t)) \dt - \alpha_2 m_2(T) \Big], \label{eq:cost_P2}
\end{align}
where $\alpha_1, \alpha_2 > 0$ are tradeoff parameters for balancing the control and scheduling costs. 
We may absorb the $\alpha_i$'s within $m_i$'s by defining new weights $\bar\lambda_{ii} = \alpha_i \lambda_{ii}$ and $\bar\lambda_{ij} = \alpha_i \lambda_{ij}$. 
Consequently, without loss of any generality, we assume $\alpha_1 = \alpha_2 = 1$ for the remainder of the article. 

Finally, the expectations in \eqref{eq:cost_P1} and \eqref{eq:cost_P2} are taken with respect to the stochasticity induced by the Brownian motion in the system dynamics, the initial state distribution, and the possible randomization in the control and scheduling policies. 

\begin{remark}\label{remark:ZSG}
    Although the game without the scheduling objective (i.e., players' objective function being \eqref{eq:pre_control_cost}) is zero-sum, the introduction of the sensing costs makes the new game (i.e., objectives \eqref{eq:cost_P1} and \eqref{eq:cost_P2}) a nonzero-sum (NZS) one.
    The new game with objectives \eqref{eq:cost_P1} and \eqref{eq:cost_P2} becomes zero-sum again under the special case $\lambda_{11} = -\lambda_{21}$, $\lambda_{22} = - \lambda_{12}$.
\end{remark}


\subsection{Information Structures and Nash Equilibrium}
Let us define the information available to the controller and the scheduler of player $i$ at time $t$ as
\begin{align*}
    {\tt I_{C_i}}(t) & := \{x(s), u_i(r), {\tt T}(t) \mid s \in {\tt T}_i(t), r \in [0,t]\} \\
    {\tt I_{C_i}}(0) & := {\tt T}(0)  \\
    {\tt I_{S_i}}(t) & := \{x(r),u_i(r),{\tt T}(t^-) \mid r \in [0,t)\} \\
    {\tt I_{S_i}}(0) & := \emptyset,
\end{align*}
where ${\tt T}(\cdot) = {\tt T}_1(\cdot) \cup {\tt T}_2(\cdot)$ is the combined scheduling instances of both players, and $t^- = \sup \{s\mid s < t\}$ denotes the time `right before' time $t$.

Let $\mu_{Ci}:{\tt I_{C_i}}\! \to \!\R^{n_{u_i}}$ denote a control policy for player $i$.
Here, $\mu_{Ci}$ lies in the space of admissible control policies defined as $\calM_{Ci}: = \{\mu_{Ci} \mid \mu_{Ci} $ is adapted to the sigma field generated by ${\tt I_{C_i}}\}$. Furthermore, let $\mu_{Si}$ denote a scheduler policy for player $i$ which lies in the space $\calM_{Si}$ of admissible scheduling policies which generate a scheduling action $\gamma_i$ according to a Bernoulli distributed random variable (the details of which are presented later in Section \ref{sec:Sensor}).
Henceforth, we will refer to the pair $\mu_i: = (\mu_{Ci}, \mu_{Si})$ as the policy of player $i$.
With the introduction of the above policy, we note that upon performing simple manipulations (subtracting $\alpha_1 \lambda_{12} n_2(t) + \alpha_2 \lambda_{22} n_2(t)$ from \eqref{eq:cost_P1} and adding $\alpha_2 \lambda_{21} n_1(t) + \alpha_1 \lambda_{11} n_1(t)$ to \eqref{eq:cost_P2}), one can arrive at a zero-sum equivalent of the NZS game. This equivalence follows because the transformation on $J_i$ is independent of self policy of $\mathrm{P}_i$, $i=1,2$.

We are now interested in deriving the Nash equilibrium policies $(\mu_1^*, \mu_2^*)$ for the players 
which satisfy

\begin{align}\label{eq:Nash_eqb}
\begin{split}
    J_1^* &:= J_1(\mu_1^*, \mu_2^*) = \min_{\mu_1 \in \calM_{C1} \times \calM_{S1}} J_1(\mu_1, \mu_2^*), \\
    J_2^* &:= J_2(\mu_1^*, \mu_2^*) = \max_{\mu_2\in \calM_{C2} \times \calM_{S2}} J_2(\mu_1^*, \mu_2).
    \end{split}
\end{align}

\section{Background and Preliminaries}\label{sec:prelim}

Before starting our main analysis, we state the solution for the unrestricted scheduling game where both players' schedulers can communicate the state $x(t)$ to their respective controllers \textit{continuously} at all times. In this case we have the zero-sum differential game with objective function \eqref{eq:pre_control_cost}-\eqref{eq:control_cost}. We assume the system model satisfies the following Assumption.
\begin{assumption}\label{Assump_1}
    The pair $(A,B_1)$ is stabilizable and the pair $(A,Q^{1/2})$ is observable.
\end{assumption}
Then, we have the following result on the unrestricted scheduling game.
\begin{theorem}\cite[Section 4.4]{bacsar2008h} \label{thm:classical}
   Suppose Assumption \ref{Assump_1} holds. Further, let $P \succeq 0$ denote the minimal positive-semidefinite solution to the generalized algebraic Riccati equation (GARE) 
\begin{align} \label{eq:ARE}
    A^{\!\top}\! P + PA + Q + P(B_2R_2^{-1}B_2^\top - B_1 R_1^{-1}B_1^\top) P = 0.
\end{align}
Then, the saddle-point control policies are unique and are given by
\begin{align} \label{eq:Nash_u}
    u_i^*(t)  = (-1)^i R_i^{-1}B_i^\top P x(t). 
\end{align}
The expected objective value of the game under this saddle-point solution is 
\begin{align}\label{eq:sec_level}
    J^* = \tr(PGG^\top).
\end{align}
\end{theorem}

The GARE \eqref{eq:ARE} admits a positive-semidefinite solution only under certain parametric assumptions on the tuple $(A, B_1, B_2, Q, R_1, R_2)$, which are necessary for the well-posedness of the game itself (see \cite[Chapter 6]{bacsar1998dynamic} or \cite[Chapter 4]{bacsar2008h} for details).

Next, we note that the saddle-point control policies \eqref{eq:Nash_u} are feedback in nature and their implementation requires continuous access to the state $x(t)$. 
However, due to the scheduling constraint of our setting, the players can schedule only \textit{intermittently} in order to balance the control performance (i.e., the integral parts in \eqref{eq:cost_P1} and \eqref{eq:cost_P2}) with the scheduling cost. 
Further, in the absence of continuous access to the state values, the optimal control policies may involve estimators (as we will demonstrate in subsequent sections) for estimating the states in-between communication instances. 

One may also notice here that the controller and scheduler information structure for player $i$ are partially nested, i.e., 
\begin{align}\label{eq:partial_nest}
    {\tt I_{C_i}}(t) \subseteq {\tt I_{S_i}}(t), ~~\forall t \in [0,T].
\end{align}
Furthermore, the actions of the controller and the scheduler are executed in a specific sequence, i.e., at any given time $t$, a decision is first made on whether to schedule a communication or not. Depending on the outcome of this decision, the controller may/may not receive new state measurements, which would alter the information available at the controller. This information (including the available history) is subsequently used to compute a control input at the controller. 

\section{Equilibrium Estimator-Controller Policy}\label{sec:Controller}

With the above background, we now proceed toward computing a Nash equilibrium policy for the linear quadratic non-zero sum (LQ-NZS) game with intermittent scheduling as formulated in Section \ref{sec:Problem_form}. In this regard, we first observe that using property \eqref{eq:partial_nest}, one may compactly rewrite \eqref{eq:Nash_eqb} along with \eqref{eq:cost_P1} and \eqref{eq:cost_P2} as:
\begin{subequations}
    \begin{align}
    J_1^* & = \min_{\mu_{S1}}\Big[ \min_{\mu_{C1}}\limsup_{T \rightarrow \infty} \frac{1}{T}\E \Big[ \int_0^T(\|x^{1,*}(t)\|_Q^2 + \|u_1(t)\|^2_{R_1} \nonumber \\
    & \hspace{.5cm} - \|u_2^*(t)\|^2_{R_2}) \dt + \lambda_{11} n_1(T) + \lambda_{12} n_{2}^*(T)\Big]\Big], \label{eq:decoupled_cost_P1}\\
    J_2^* & = \max_{\mu_{S2}}\Big[ \max_{\mu_{C2}}\limsup_{T \rightarrow \infty} \frac{1}{T}\E \Big[ \int_0^T(\|x^{2,*}(t)\|_Q^2 + \|u_1^*(t)\|^2_{R_1} \nonumber \\
    & \hspace{.5cm} - \|u_2(t)\|^2_{R_2}) \dt - \lambda_{22} n_2(T) - \lambda_{21} n_{1}^*(T)\Big]\Big]. \label{eq:decoupled_cost_P2}
\end{align}
\end{subequations}
where $n_i^* = {\tt card}({\tt T}_i(t))$ under the equilibrium scheduling policy of player $i$ and $u_i^*(\cdot)$ is the equilibrium control policy. 
The states $x^{i,*}$ for $i= 1$ and $2$ in \eqref{eq:decoupled_cost_P1} and \eqref{eq:decoupled_cost_P2}, respectively, follow the dynamics \eqref{eq:dyn} under the equilibrium control policy of player $-i$. 
That is,
\begin{align*}
    \dx^{i,*}(t) & = [Ax^{i,*}(t) + B_iu_i(t) + B_{-i}u_{-i}^*(t)] \dt + G \dW(t).
\end{align*}

The above reformulation suggests that we can now solve for the Nash equilibrium policies in two stages by first computing the equilibrium control policies under fixed scheduler policies and consequently computing the equilibrium scheduler policies.

Let us begin by computing the Nash control policies in this section. In this regard,
we first use completion of squares to rewrite \eqref{eq:pre_control_cost} as:
\begin{align}
    J(\mu_1,\mu_2) &= J^*\! + \! \limsup_{T \rightarrow \infty} \frac{1}{T}\E \Big[ \int_0^T \!\!\|u_1(t) + R_1^{-1}B_1^\top P x(t)\|_{R_1}^2 \nonumber \\
    & - \|u_2(t) - R_2^{-1}B_2^\top  P x(t)\|_{R_2}^2 \dt \Big],
\end{align}
where we recall that $J^*$ is defined in \eqref{eq:sec_level}.
This then leads us to the following result on the equilibrium control policies for both players.

\begin{theorem}\label{thm:Nash_control}
Suppose that the hypotheses of Theorem \ref{thm:classical} hold. Further suppose there exists a nonpositive definite solution $\tilde P_{11}$ to the algebraic Riccati equation (ARE)
\begin{align}\label{ARE_P2}
    \tilde P_{11} A + A^\top \tilde P_{11} -Q - \tilde P_{11} B_2 R_2^{-1} B_2^\top \tilde P_{11} = 0.
\end{align}
Then, for any fixed scheduler policies $\mu_{Si}$ for $i=1,2$, the following holds:
    \begin{enumerate}
        \item The equilibrium control policy is unique and is given as:
        \begin{align}\label{eq:eqb_control}
            u_i^*(t) = \mu^*_{Ci}({\tt I_{C_i}}(t)) = (-1)^i R_i^{-1} B_i^\top P\hat{x}_i^*(t),
        \end{align}
        where $P$ is the solution to the GARE in \eqref{eq:ARE} and the  saddle-point equilibrium is attained at $(\hat{x}_1^*(\cdot), \hat{x}_2^*(\cdot))$ where
        \begin{align}
            \hat x_i^*(t) = \bbE[x(t) \mid {\tt I_{C_i}}(t)],~ i = 1,2.
        \end{align}
        \item Define  $\Lambda_i:= PB_i R_i^{-1} B_i^\top P$ for $i=1,2$. Then, the saddle point estimators $(\hat{x}_1^*(\cdot), \hat{x}_2^*(\cdot))$ satisfy the following set of jump differential equations:
        \begin{align}
            &\dot{\hat x}^*_1(t)  \!= \! [A - P^{-1} \Lambda_1 \!+\! P^{-1} \Lambda_2] \hat x_1^*(t), ~ t \in [0,\infty) \!\setminus \! {\tt T}_1 \nonumber \\
            & \hat x_1^*(t_\ell)  = x(t_\ell), \qquad \forall ~~ t_{\ell} \in {\tt T}_1. \label{eq:eqb_xhat1}
        \end{align}
        \begin{align}
            &\dot{\hat x}^*_2(t)  \!=\! [A - P^{-1} \Lambda_1 \!+\! P^{-1} \Lambda_2] \hat x_2^*(t), ~ t \in [0,\infty) \!\setminus \! {\tt T}_2 \nonumber \\
            &\hat x_2^*(t_{\ell'})  = x(t_{\ell'}), \qquad \forall~~ t_{\ell} \in {\tt T}_2. \label{eq:eqb_xhat2}
        \end{align}
    \end{enumerate}
\end{theorem}
\begin{proof}
    The proof is provided in Appendix I.
\end{proof}

Theorem \ref{thm:Nash_control} provides sufficient conditions for the existence of a \textit{unique} equilibrium pair of control policies for the players.
The existence of solution to the ARE  \eqref{ARE_P2} is to ensure well-posedness of the game by ensuring that player $\Pb$ cannot drive the objective function  \eqref{eq:cost_P2} to infinity. 
It is noteworthy that such well-posedness conditions are necessary for the existence of equilibrium controllers pair in LQ (non) zero-sum games. 
For instance, even in the standard case of continuous communication (i.e., \Cref{thm:classical}) such conditions are imposed by the GARE \eqref{eq:ARE}. 
As we are studying an \textit{intermittent communication} scenario in \eqref{thm:Nash_control}, we need the additional condition \eqref{ARE_P2}. Finally, under Assumption \ref{Assump_1}, the closed-loop matrix $\tilde A:= A - P^{-1} \Lambda_1 \!+\! P^{-1} \Lambda_2$ is Hurwitz (cf. \cite[Theorem 9.7]{bacsar2008h}.

We further note that, LQ games with intermittent communication has been studied \cite{bacsar1991optimum, bacsar1995minimax} under the `sampled-data' framework, where the scheduling instances (i.e., the sets ${\tt T}_i$) are fixed a priori. 
In those cases, similar (but less restrictive) conditions have been discussed for the finite horizon \cite[equation 2.7]{bacsar1991optimum}, \cite[Theorem 3.1]{bacsar1995minimax} and the infinite horizon cases \cite[Theorem  2.2]{bacsar1991optimum}, \cite[equation 3.9]{bacsar1995minimax}, which are necessary for the existence of the equilibrium controller pair. 
In our case, the scheduling instances are design parameters and no upperbound on the inter-sampling instances is imposed (e.g., $t_{\ell+1,i} - t_{\ell,i} < \tau$ for all $l$ and i), which necessitates a stricter condition than those in \cite{bacsar1991optimum, bacsar1995minimax}.\footnote{Our condition \eqref{ARE_P2} may be viewed as a limiting case of \cite[Theorem 2.2]{bacsar1991optimum} where the inter-sampling duration (i.e., $t_{k+1} - t_{k}$) approaches infinity. It can also be viewed as the condition that ensures concavity of the maximization problem when the minimizer uses any open-loop policy \cite{bacsar1998dynamic,bacsar2008h}.}

A more recent work \cite{maity2023efficient} has considered a \textit{finite-horizon} intermittent communication framework with the communication instances being part of the optimization variables.
In that work, a finite-horizon equivalent to the ARE \eqref{ARE_P2} determined the communication instances.\footnote{
The finite-horizon equivalent is the Riccati differential equation:
\begin{align}
   \dot{\tilde{P}}_{11} + \tilde P_{11} A + A^\top \tilde P_{11} -Q - \tilde P_{11} B_2 R_2^{-1} B_2^\top \tilde P_{11} = 0.
\end{align}
} 
More precisely, the \textit{conjugate points} (i.e., finite escape times) of the finite-horizon equivalent of \eqref{ARE_P2} determined the communication instances; see \cite[Theorem~2]{maity2023efficient} for details. 
At this point, it is noteworthy that only $\Pa$ was restricted to intermittent communication in \cite{maity2023efficient} and $\Pb$ had continuous access to the state---a different problem set up than the one considered here.

We further remark that Theorem \ref{thm:Nash_control} is different from the results derived in \cite{bacsar1991optimum,bacsar1995minimax} in two aspects: 

(1) We consider asynchronous scheduling for both players, i.e., the sets of scheduling instants for the players can be different, which is in contrast to the setup of the aforementioned works where synchronous (and pre-defined) scheduling was considered, and 

(2) Our scheduling policy (as we will formalize later) is parametrized within the class of randomized policies thereby leading to the possibility of an unbounded scheduling interval. This thus requires that the solution to the ARE in \eqref{ARE_P2} exists as opposed to the differential Riccati equations proposed in \cite{bacsar1991optimum,bacsar1995minimax} with fixed scheduling instants.


\subsection{Error Dynamics and  Objective Functions}
Let us define the estimation error for player $i$ as $e_i(t) = x(t) - \hat x^*_i(t)$. Then, using the control policy from Theorem \ref{thm:Nash_control}, the error dynamics satisfy the following jump differential equations:
For $t \in [0,\infty) \setminus {\tt T}_1$,
\begin{subequations}
    \begin{align}
    \de_1(t) & = [A + P^{-1} \Lambda_2]e_1(t) \dt - P^{-1} \Lambda_2 e_2(t)\dt + G \dW(t), \nonumber\\
    e_1(t_\ell) & = 0,  \qquad \forall ~~ t_{\ell} \in {\tt T}_1, \label{eq:error1_dyn}
    \end{align}
     and for $t \in [0,\infty) \setminus {\tt T}_2$
    \begin{align}
    \de_2(t) & = [A - P^{-1} \Lambda_1]e_2(t) \dt + P^{-1} \Lambda_1 e_1(t)dt + G \dW(t), \nonumber \\
    e_2(t_{\ell'})&  = 0,  \qquad \forall ~~ t_{\ell'} \in {\tt T}_2. \label{eq:error2_dyn}
\end{align}
\end{subequations}
Let $e(t):= [e_1^\top(t)~ e_2^\top(t)]^\top$. Then, we can compactly write the set of equations \eqref{eq:error1_dyn}-\eqref{eq:error2_dyn} as:
\begin{align}\label{eq:err_combined}
    & \de(t)  = \bar A e(t) \dt + \bar G \dW(t),~\forall ~ t \in [0,\infty) \setminus ({\tt T}_1 \cup {\tt T}_2), \nonumber \\
    & e_1(t_\ell)  = 0   \hspace{2.85 cm} \forall ~ t_{\ell} \in {\tt T}_1, \nonumber \\
    & e_2(t_{\ell'})  = 0 \hspace{2.75 cm} \forall ~ t_{\ell'} \in {\tt T}_2,
\end{align}
where 
\begin{align*}
    \bar A = \begin{bmatrix}
        A + P^{-1} \Lambda_2 & -P^{-1} \Lambda_2 \\
        P^{-1} \Lambda_1 & A - P^{-1} \Lambda_1
    \end{bmatrix} \text{ and }
        \bar G = \begin{bmatrix}
            G \\ G
        \end{bmatrix}.
\end{align*}
Further, the `closed-loop' objective values can be obtained by substituting the saddle-point controller policies in the objective functions \eqref{eq:cost_P1} and \eqref{eq:cost_P2} as:

\begin{align}
J_1(\mu_{S1}, \mu_{S2}) & \!=\!  J^* \!\!+ \limsup_{T \rightarrow \infty} \frac{1}{T}\E \Big[ \int_0^T \!\!\|e_1(t)\|^2_{\Lambda_1} \!-\! \|e_2(t)\|^2_{\Lambda_2} \dt \nonumber \\
& + \lambda_{11} n_1(T) + \lambda_{12} n_2(T)\Big], \label{eq:closed_loop_cost1}\\
J_2(\mu_{S1}, \mu_{S2}) & \!=\!  J^* \!\! + \limsup_{T \rightarrow \infty} \frac{1}{T}\E \Big[ \int_0^T \!\! \|e_1(t)\|^2_{\Lambda_1} \!-\! \|e_2(t)\|^2_{\Lambda_2} \dt \nonumber \\
& - \lambda_{21} n_1(T) - \lambda_{22} n_2(T)\Big], \label{eq:closed_loop_cost2}
\end{align}
where we recall from \eqref{eq:sec_level} that $J^* = \tr(PG G^\top)$ is the game value under the continuous communication case.


\begin{remark}
    Since $\lambda_{ii}>0$, an admissible scheduling policy must satisfy $\limsup_{T \rightarrow \infty} \E[n_i(T)/T] < +\infty$.
    In other words, the scheduling policies will be Zeno-free,\footnote{
    A Zeno behavior is where an uncountable number of switching/trigger-ing/scheduling occurs in a finite interval.
     }---a requirement that is often to be independently verified in intermittent-feedback (e.g., self/event-triggered) controls but is automatically satisfied in this work due to the problem structure.
\end{remark}

The derivation of the (Nash) equilibrium control policies is now complete, which holds for any fixed scheduling policies. 
In the following section, we will focus on constructing the equilibrium scheduling policies for both players. To facilitate this, we will work with a time-discretized version of the continuous system, which allows for a tractable analysis and facilitates the development of efficient algorithms for numerical computations. 
This approach is particularly advantageous compared to the continuous-time domain, where it is often difficult to provide explicit methods for computing the scheduling instants, and where solutions tend to be limited to existential results rather than practical algorithms \cite{baras1989optimal}. 
Additionally, we will adopt a parametrized class of scheduling policies, which simplifies the task of deriving explicit expressions for scheduler policy computation. Although the considered class is open-loop in nature, it is particularly helpful in circumventing the complex behaviour of feedback policies, where computing scheduling policies even within the case of a single-player optimal control problem can be challenging \cite{soleymani2021value}.

\section{Equilibrium Scheduler Policy}\label{sec:Sensor}
\subsection{Discrete-time Formulation}
To entail a tractable analysis of the scheduling policy computation, we begin by discretizing the system dynamics \eqref{eq:error1_dyn}-\eqref{eq:error2_dyn} and the objective functions \eqref{eq:closed_loop_cost1}-\eqref{eq:closed_loop_cost2} using Euler approximation.
%
%
For a fixed discretization interval $h$, let us denote the discrete time index by $k$. 
Further, let us define the set of discrete communication instants by 
\begin{align}\label{eq:disc}
    {\tt S}_i(k) := {\tt T}_i(kh).
\end{align}
We also restrict the possible communication instances (i.e., $t_{\ell, i}$'s) to be multiples of $h$, which results in ${\tt S}_i (k) \subseteq \{0,1,\ldots, k\}$.

Then, using \eqref{eq:err_combined}, one may arrive at the following error dynamics in discrete time as:
\begin{align}\label{eq:err_combined_disc}
    e_{k+1}^- & = e^{\bar A h} e_k + \int_0^h e^{\bar A (h-s)} \bar G \dW(s), \nonumber \\
    e_k &= \begin{cases}
        e_k^- &\forall k \notin {\tt S}_1 \cup {\tt S}_2, \\~\\
         \begin{bmatrix}
        0 & 0 \\
        0 & I
    \end{bmatrix} e_\ell^-, &\forall k \in {\tt S}_1 \cap {\tt S}_2^c, \\~\\
        \begin{bmatrix}
        I & 0 \\
        0 & 0
    \end{bmatrix} e_k^-, &~\forall k \in {\tt S}_2 \cap {\tt S}_1^c,\\~\\
    0, &~\forall k \in {\tt S}_1 \cap {\tt S}_2,
    \end{cases} 
\end{align}
where $e_k = [e_{1,k}^\top~ e_{2,k}^\top]^\top$ and $e_{i,k} = e_i(kh), ~i=1,2$.
Here, $e_k^-$ denotes the estimation error at the controllers before any communication is scheduled. 
After the communication is scheduled at time $k$ (if any), the error is updated and is denoted as $e_k$.

Let us consider the following class of parametrized (stationary) randomized policies as:
\begin{subequations}
    \begin{align}
    \gamma_{1,k} & = \mu_{S1}({\tt I_{S1}}(k)) \sim \text{Bernoulli}(1-p) \\
    \gamma_{2,k} & = \mu_{S2}({\tt I_{S2}}(k)) \sim \text{Bernoulli}(1-q),
\end{align}
\end{subequations}
where $0\leq p,q \leq 1$, and $\gamma_{i,k} \in \{0,1\}$ denotes the scheduling action of player $i$. In this section, we are interested in finding Nash equilibrium scheduling policies for both players within the aforementioned class, i.e., finding the equilibrium pair ($p^*, q^*$). 

    

Now, we will express the objective functions \eqref{eq:closed_loop_cost1}-\eqref{eq:closed_loop_cost2} in discrete-time format as well.
To that end, note that the possible communication instances are $\{kh\}_{k\in \mathbb N}$, and therefore, for any interval $[kh, (k+1)h)$, we have from \eqref{eq:err_combined} that
\begin{align} \label{eq:new_error}
    e(t) =  e^{\bar A (t-kh)} e_k + \int_{kh}^t e^{\bar A (t-s)} \bar G \dW(s),
\end{align}
for all $t \in [kh, (k+1)h)$. 
Consequently, we may further write
\begin{align} \label{eq:e_simplification}
     & \E \left[\int_{kh}^{(k+1)h} \!\!    \|e_1(t)\|^2_{\Lambda_1} \!-\! \|e_2(t)\|^2_{\Lambda_2} \dt \right]  \nonumber \\
     & \quad = \E \left[  \int_{kh}^{(k+1)h} \!\! \|e(t)\|^2_{\Lambda}  \dt \right] 
      \overset{(\dagger)}{=}  \| e_k\|^2_{\bar{\Lambda}(h)} + \varphi(h),
\end{align}
where $\Lambda = {\tt Diag}[\Lambda_1, -\Lambda_2]$ and where $(\dagger)$ follows from substituting \eqref{eq:new_error} for $e(t)$ and using the fact that $\dW(s)$ is independent of $e_k$ for all $s\ge kh$. 
Here,
\begin{subequations}
\begin{align}
   &\bar{\Lambda}(h) = \int_{0}^{h} e^{\bar A^{\!\top} s} \Lambda e^{\bar A s} \ds, \\
   &\varphi(h) =  \int_{t=0}^{h}  \int_{s=0}^{t} \tr\left( e^{\bar{A}^{\! \top} (t-s) } G^\top \Lambda G e^{\bar{A} (t-s) } \right)\ds \dt 
\end{align}    
\end{subequations}
in \eqref{eq:e_simplification}. 
Therefore, we may further write 
\begin{align}
    \limsup_{T \rightarrow \infty} & \frac{1}{T}\E \Big[ \int_0^T \!\!\|e(t)\|^2_{\Lambda} \dt \Big] = \limsup_{m \rightarrow \infty} \frac{1}{mh}\E \Big[ \int_0^{mh} \!\!\|e(t)\|^2_{\Lambda}  \dt \Big] \nonumber \\
    &= \limsup_{m \rightarrow \infty} \frac{1}{mh}\E \Big[ \sum_{k=0}^{m-1} \int_{kh}^{(k+1)h} \!\!\|e(t)\|^2_{\Lambda}  \dt \Big] \nonumber \\
    & \overset{\eqref{eq:e_simplification}}{=}  \limsup_{m \rightarrow \infty} \frac{1}{m}\E \Big[ \sum_{k=0}^{m-1} \| e_k\|^2 _{\frac{\bar{\Lambda}(h)}{h}} \Big] + \frac{\varphi(h)}{h}.
\end{align}
Note that the term $\frac{\varphi(h)}{h}$ is independent of the scheduling policy and is simply a result of rewriting the integration into a summation form. 
It is also noteworthy that $\lim_{h\to 0} \frac{\varphi(h)}{h} = 0$ and $\lim_{h\to 0} \frac{\bar{\Lambda}(h)}{h} = \Lambda$. 
The objective functions \eqref{eq:closed_loop_cost1}-\eqref{eq:closed_loop_cost2} are now rewritten in the discrete-time format as:
\begin{align}
J_1(\mu_{S1}, \mu_{S2}) & =  \tilde{J}^* + \limsup_{m \rightarrow \infty} \frac{1}{m}\E \Big[ \sum_{k=0}^{m-1} \big( \|e_{k}\|^2_{\frac{\bar{\Lambda}(h)}{h}}  \nonumber \\
& \hspace{2.5cm}+ \lambda_{11} \gamma_{1,k} + \lambda_{12} \gamma_{2,k} \big) \Big] \label{eq:closed_loop_cost_disc1} \\
J_2(\mu_{S1}, \mu_{S2}) & =   \tilde{J}^* + \limsup_{m \rightarrow \infty} \frac{1}{m}\E \Big[ \sum_{k=0}^{m-1} \big( \|e_{k}\|^2_{\frac{\bar{\Lambda}(h)}{h}} \nonumber \\
& \hspace{2.5cm} - \lambda_{21} \gamma_{1,k} - \lambda_{22} \gamma_{2,k} \big)\Big],\label{eq:closed_loop_cost_disc2}
\end{align}
where $\tilde{J}^* = J^* + \frac{\varphi(h)}{h} $.

\subsection{Equilibrium Scheduling Policy}
Our aim is to find a Nash scheduling policy pair $(p^*,q^*)$ by optimizing the objective functions \eqref{eq:closed_loop_cost_disc1}-\eqref{eq:closed_loop_cost_disc2} subject to error dynamics \eqref{eq:err_combined_disc}, satisfying the following set of inequalities:
\begin{align*}
    J_1(p^*,q^*) \leq J_1(p,q^*) , &~\forall p \in [0,1] \\
    J_2(p^*,q^*) \geq J_2(p^*,q), &~\forall q \in [0,1].
\end{align*}

In this regard, let us define the error covariance matrix as:
\begin{align}
    \Sigma_k = \mathbb E [e_k e_k^\top] =: \begin{bmatrix}
\Sigma_{11,k} & \Sigma_{12,k} \\
\Sigma_{12,k}^\top & \Sigma_{22,k}
\end{bmatrix},
\end{align}
where the expectation is taken with respect to the randomness induced by the system noise and also the scheduling policy.
Consequently, we state the following proposition, which, as we will show later, will aid in reformulating the equilibrium computation problem as described above into solution of a generalized Sylvester-type equation.

\begin{proposition}\label{Prop:cov_dynamics}
    The evolution dynamics of the error covariance matrix $\Sigma_k$ is given by the following equation:
    \begin{align}\label{eq:error_covariance}
        \Sigma_{k+1} & = A_1(p,q) \Sigma_k A_1(p,q)^\top + A_2(p,q) \Sigma_k A_2(p,q)^\top \nonumber \\
    & ~~~+ A_3(p,q) \Sigma_k A_3(p,q)^\top + G(p,q) \\
    A_1(p,q) & := \begin{bmatrix}
        p \Phi_{11}(h) & p \Phi_{12}(h) \\
        q \Phi_{21}(h) & q \Phi_{22}(h)
    \end{bmatrix} \nonumber \\
    A_2(p,q) & := \sqrt{p(1-p)}\begin{bmatrix}
        \Phi_{11}(h) & \Phi_{12}(h) \\
        0 & 0
    \end{bmatrix} \nonumber \\
    A_3(p,q) & := \sqrt{q(1-q)}\begin{bmatrix}
        0 & 0 \\
        \Phi_{21}(h) & \Phi_{22}(h)
    \end{bmatrix} \nonumber
    \end{align}
    \begin{align*}
    G(p,q) & := \begin{bmatrix}
        p \tilde G_1 & pq \tilde G_2 \\
        pq \tilde G_2^\top & q \tilde G_3
    \end{bmatrix}.
\end{align*}
\end{proposition}
\begin{proof}
    The proof is presented in Appendix~II.
\end{proof}

Consequently, by defining $\bar p:= 1-p$ and $\bar q:= 1-q$, one may also express the objective functions \eqref{eq:closed_loop_cost_disc1}-\eqref{eq:closed_loop_cost_disc2} in terms of $\Sigma_k$ as:
\begin{subequations}\label{eq:Sigma_cost}
    \begin{align}
    J_1(\mu_{S1}, \mu_{S2}) & =  \tilde J^* \!+\! \limsup_{m \rightarrow \infty} \frac{1}{m}\!\sum_{k=0}^{m-1} \!\big[ \tr(\tilde \Lambda \Sigma_k) \!+\! \lambda_{11} \bar p \!+ \! \lambda_{12} \bar q\big] \label{eq:Sigma_cost1} \\
J_2(\mu_{S1}, \mu_{S2}) & =  \tilde J^* \!+\! \limsup_{m \rightarrow \infty} \frac{1}{m} \sum_{k=0}^{m-1} \! \big[ \tr(\tilde \Lambda \Sigma_k) \! - \! \lambda_{21} \bar p\! - \!\lambda_{22}\bar q\Big],\label{eq:Sigma_cost2}
\end{align}
\end{subequations}
where $\tilde \Lambda = \frac{\Lambda(h)}{h}$.
Thus our reformulated objective is to find a Nash equilibrium $(p^*,q^*)$ by minimizing (resp. maximizing) the objective \eqref{eq:Sigma_cost1} (resp. \eqref{eq:Sigma_cost2}) subject to the equation \eqref{eq:error_covariance}.

Let us define $\Sigma_\infty \triangleq \limsup_{m \rightarrow \infty} \frac{1}{m}\sum_{k=0}^{m-1}\Sigma_k$. Since the limit superior of any sequence is always well-defined, $\Sigma_\infty$ is either bounded or may take unbounded values for at least one element $(\Sigma_\infty)_{ij}$. Thus, the objective functions in \eqref{eq:Sigma_cost} can be written as:
\begin{subequations}\label{eq:cost_sig_inf}
    \begin{align}
    J_1(\mu_{S1}, \mu_{S2}) & = \tilde J^* +  [ \tr(\tilde\Lambda \Sigma_\infty) + \lambda_{11} \bar p + \lambda_{12} \bar q] \label{eq:aa}\\
    J_1(\mu_{S1}, \mu_{S2}) & =  \tilde J^* +  [ \tr(\tilde \Lambda \Sigma_\infty) - \lambda_{21} \bar p - \lambda_{22} \bar q], \label{eq:ab}
\end{align}
\end{subequations}
and $\Sigma_\infty$ satisfies the equation:
\begin{align}\label{eq:ss_cov}
    \Sigma_\infty & = A_1(p,q) \Sigma_\infty A_1(p,q)^\top + A_2(p,q) \Sigma_\infty A_2(p,q)^\top \nonumber \\
    & \hspace{2cm} + A_3(p,q) \Sigma_\infty A_3(p,q)^\top + G(p,q).
\end{align}
The equilibrium scheduling policy design now amounts to finding a steady state solution $\Sigma_\infty$ optimizing the objective functions \eqref{eq:cost_sig_inf} subject to the equation \eqref{eq:ss_cov}.

To that end, we first state the following result which discusses conditions for the existence of a unique bounded solution to \eqref{eq:ss_cov}. Let us define the linear operators $\Pi_1(X):= X$ and $\Pi_2(X) = -\sum_{i=1,2,3} A_i(p,q) X A_i(p,q)^\top$. Further, let $C:= G(p,q)$. Then, the solution to \eqref{eq:ss_cov} can be given in terms of Neumann series \cite{bowers2014introductory} of the above linear operators as presented in the following theorem.

\begin{theorem}[\!\!\cite{jarlebring2018krylov}]\label{thm:approximation}
    Suppose $\rho(\Pi_1^{-1} \Pi_2) <1$, where $\rho(T):=\sup \{|\lambda|:\lambda \in S(T) \}$ as the (operator) spectral radius of $T$ and $S(T)$ denotes the set of eigenvalues of $T$. Then, the following statements hold:
    \begin{enumerate}
        \item The unique bounded solution to \eqref{eq:ss_cov} is obtained as
    \begin{align*}
        \Sigma_\infty^* & = \sum_{j=0}^\infty Y_j, \\
        Y_{j+1} & = -\Pi_1^{-1}(\Pi_2(Y_j)), j \geq 0, ~~ Y_0 = \Pi_1^{-1}(C).
    \end{align*} 
    
    \item Further, if we define $\Sigma_\infty^{(t_p)}:= \sum_{j=0}^{t_p} Y_j$ for a positive truncation parameter $t_p$, then we have that the truncation error $\|\Sigma_\infty^* - \Sigma_\infty^{(t_p)}\|$ can be bounded as
    \begin{align}\label{eq:trunc_Sigma}
        \|\Sigma_\infty^* - \Sigma_\infty^{(t_p)}\| \leq \|\Pi_1^{-1}(C)\| \frac{[\rho(\Pi_1^{-1} \Pi_2)]^{t_p+1}}{1- \rho(\Pi_1^{-1} \Pi_2)}.
    \end{align}
    \end{enumerate}
\end{theorem}
Briefly, the above theorem characterizes the steady-state solution to \eqref{eq:ss_cov} in terms of sums of infinite powers of the variables $p$ and $q$. Additionally, \eqref{eq:trunc_Sigma} provides an approximation of the steady state covariance $\Sigma^*_\infty$ using the truncation parameter $t_p$ and finite powers of the polynomial variables $p$ and $q$.

\subsection{Algorithmic Computation of the Equilibrium Policies}
As noted earlier, computing a Nash equilibrium requires optimizing over the set of objective functions \eqref{eq:cost_sig_inf} subject to  \eqref{eq:ss_cov}. 
Furthermore, the steady state solution to \eqref{eq:ss_cov} constitutes an (infinite) series of polynomial terms  in the variables $p$ and $q$, which makes it difficult to carry out an analytical computation of the Nash scheduling policies; thus, the analysis here is rather focused on the development of efficient algorithms.
In the following, we provide two algorithms for computing the Nash policies---one of them involves an exhaustive search and the other one involves an iterative method. 
However, fundamentally, both the algorithms are constructed using the principle of {\tt BestResponse}.
\begin{algorithm}[h]
	\caption{{\tt BestResponse}($i,p_{-i},\eta$) }
 \label{alg:BR}
	\begin{algorithmic}[1]
    \STATE {\textbf{Input: } $i$} \hfill \# Player id: $\Pa,\Pb$ 
    \STATE {\textbf{Input: } $p_{-i}$} \hfill \# Opponent's scheduling probability
    \STATE {\textbf{Input: } $\eta$} \hfill \# Gradient-descent step-size 
    \STATE {\textbf{Parameters:} $\kappa $ \hfill \# stopping tolerance}
    \STATE Initialize $p \in [0,1]$ 
    \IF{$i = \Pa$} \hfill \# computing best response for $\Pa$
    \STATE $ g \gets \tr(\tilde{\Lambda} \frac{\partial \Sigma_\infty}{\partial p} (p, p_{-i})) + \lambda_{11} $  \# gradient of objective \eqref{eq:aa}
    \STATE $p^{\rm new} \gets p - \eta g $ \hfill \# gradient-descent 
    \STATE $p^{\rm new} \gets \max(\min(p^{\rm new},1), 0)$ \hfill \# projection in $[0,1]$
    \IF {$|p - p^{\rm new}| > \kappa$}
    \STATE $p\gets p^{\rm new}$ 
    \STATE go to line 7 \hfill \# continue projected gradient-descent
    \ELSE 
    \STATE go to line 27 \hfill \# stop gradient-descent
    \ENDIF
    \ELSE  \hfill \# computing best response for $\Pb$
    \STATE $ g \gets \tr(\tilde{\Lambda} \frac{\partial \Sigma_\infty}{\partial p} (p_{-i},p)) - \lambda_{22} $ \# gradient of objective \eqref{eq:ab}
    \STATE $p^{\rm new} \gets p + \eta g $ \hfill \# gradient-ascent 
    \STATE $p^{\rm new} \gets \max(\min(p^{\rm new},1), 0)$ \hfill \# projection in $[0,1]$
    \IF {$|p - p^{\rm new}| > \kappa$}
    \STATE $p\gets p^{\rm new}$ 
    \STATE go to line 16 \hfill \# continue projected gradient-descent
    \ELSE 
    \STATE go to line 27 \hfill \# stop gradient-descent
    \ENDIF
    \ENDIF
    \STATE \textbf{return} $p$
    \end{algorithmic}
\end{algorithm}

{\tt BestResponse} is computed in \Cref{alg:BR}, which involves a projected gradient based algorithm.
We note that the {\tt BestResponse}$(\cdot)$ function (\Cref{alg:BR}) requires the computation of the gradients $\frac{\partial \Sigma_\infty}{\partial p} (p, p_{-i})$ (line 7) and $\frac{\partial \Sigma_\infty}{\partial p} (p_{-i},p)$ (line 17), where $\Sigma_\infty(p,q)$ satisfies \eqref{eq:ss_cov}.
In Appendix~III we demonstrate how these derivatives can be computed by solving certain Sylvester-type equations, which can be solved (approximately) using \Cref{thm:approximation}.

\subsubsection{Exhaustive Search Method} In this case, for every value of $q \in [0,1]$ (i.e., $\Pb$'s strategy) we compute the best response for $\Pa$ (i.e.,  {\tt BestResponse}($\Pa,q,\eta_1$)) using \Cref{alg:BR}. 
Let $p^*(q)$ denote the best response value returned by {\tt BestResponse}($\Pa,q,\eta_1$).
Similarly, for every $p\in [0,1]$, let $q^*(p)$ denote the best response value returned by {\tt BestResponse}($\Pb,p,\eta_2$).
Afterwards, we plot $(p, q^*(p))$ for all $p\in [0,1]$ and $(p^*(q), q)$ for all $q \in [0,1]$ and search for the intersection of these two graphs; see Fig.~\ref{Fig:BR_plot} for an illustration. 
The intersection is a Nash equilibrium.
The formal algorithm is presented in \Cref{alg:exhaustive}. 
\begin{algorithm}[h]
	\caption{Exhaustive Search for computing NE}
 \label{alg:exhaustive}
	\begin{algorithmic}[1]
        \STATE {\textbf{Parameter:} $\varepsilon$ \hfill \# coarseness of exhaustive search}
        \STATE $p = 0: \varepsilon: 1$ \hfill \# discretize the values of $p$
        \STATE $q = 0: \varepsilon: 1$ \hfill \# discretize the values of $q$
		\STATE {Initialize: $p^{(0)},q^{(0)} \in [0,1]$}
        \FOR {$p = 0: \varepsilon: 1$}
        \STATE $q^*(p) \gets $ {\tt BestResponse}($\Pb,p,\eta_2$)
        \ENDFOR
        \FOR {$q = 0: \varepsilon: 1$}
        \STATE $p^*(q) \gets $ {\tt BestResponse}($\Pa,q,\eta_1$)
        \ENDFOR
        \STATE NE Pairs $\gets \{(p, q^*(p))\}_{p\in \{0:\varepsilon:1\}} \cap \{ (p^*(q), q)\}_{q\in \{0:\varepsilon:1\}} $ \flushright \# set of all Nash pairs
	\end{algorithmic}
\end{algorithm}

\begin{remark}
    The parameter $\varepsilon$ in \Cref{alg:exhaustive} needs to be chosen sufficiently small to have a non-empty intersection in line~11 (assuming that at least one Nash equilibrium exists). 
    Furthermore, even with arbitrarily small $\varepsilon$ it may still not be possible to find a pair $(p^*, q^*) \in \{(p, q^*(p))\}_{p\in \{0:\varepsilon:1\}} \cap \{ (p^*(q), q)\}_{q\in \{0:\varepsilon:1\}}$ (say, for instance, irrational values of $p^*$ and $q^*$). 
    In that case, one may seek a pair $(p', q')$ such that there exists $(p^1, q^1) \in \{(p, q^*(p))\}_{p\in \{0:\varepsilon:1\}}$ with $|p'- p^1|+|q' - q^1| < \sigma$ and another pair  $(p^2, q^2) \in \{(p^*(q),q)\}_{q\in \{0:\varepsilon:1\}}$ with $|p'- p^2|+|q' - q^2| < \sigma$ for sufficiently small tolerance parameter $\sigma$. 
\end{remark}

\subsubsection{Iterative Search Method}
Next, we provide another algorithm (Algorithm~\ref{alg:Nash}) for computing (an approximate) Nash scheduling policy pair for both players.
The algorithm uses a projected gradient descent/ascent to iteratively compute the best response of each player for a given policy (i.e., value of $p$ or $q$) of the opponent.
If the algorithm converges, then it outputs a Nash pair ($p^*, q^*$).
%
%
\begin{algorithm}[h]
	\caption{Iterative Method for computing NE}
 \label{alg:Nash}
	\begin{algorithmic}[1]
        \STATE {\textbf{Input:} $\varepsilon$ \hfill \# tolerance parameter}
        \STATE {\textbf{Input:} { $\eta_1,\eta_2$} \hfill \# Gradient-descent step sizes}
		\STATE {Initialize: $p^{(0)},q^{(0)} \in [0,1]$}
        \STATE {$p^{(1)} \gets {\tt BestResponse}(\Pa,q^{(0)},\eta_1)$}
        \STATE {$q^{(1)} \gets {\tt BestResponse}(\Pb,p^{(1)},\eta_2)$}
        \STATE $k \gets 1$ \hfill \# iteration counter
        \WHILE{$\max\{|p^{(k)} - p^{(k-1)}|, |q^{(k)} - q^{(k-1)}| \} {>} \varepsilon$} 
        \STATE {$p^{(k+1)} \gets {\tt BestResponse}(\Pa,q^{(k)},\eta_1)$}
        \STATE {$q^{(k+1)} \gets {\tt BestResponse}(\Pb,p^{(k+1)},\eta_2)$}
        \STATE $k \gets k+1$ \hfill \# counter increment
        \ENDWHILE
        \STATE \textbf{Output:} $(p^*, q^*) \gets (p^{(k)}, q^{(k)}).$  \# approximate Nash pair
	\end{algorithmic}
\end{algorithm}
%

Finally, we note that since there may, in general, exist multiple Nash equilibria, one may wish to choose the `best' among those \textit{if} they are comparable. A Nash policy NE1 is better than another Nash policy NE2 if NE1 leads to a strictly better objective value for at least one of the players, and yields no less value for the other player, compared to when NE2 is employed by both players. Of course, two Nash equilibria may not even be comparable. Thus, one may run the above Algorithm for multiple randomly chosen initial conditions to find the best Nash scheduling policies (if possible).

Our discussion on the computation of both the Nash scheduling policy and control policy is now complete, and we proceed to simulate some examples to verify the validity of the theoretical results.
\section{Performance Evaluation}\label{sec:numSims}
In this section, we numerically assess the merits of the proposed Algorithm \ref{alg:Nash} and evaluate its performance on two numerical examples.

\begin{figure}[!b]
\vspace{-4mm}
	\centering
	\includegraphics[width=\columnwidth ]{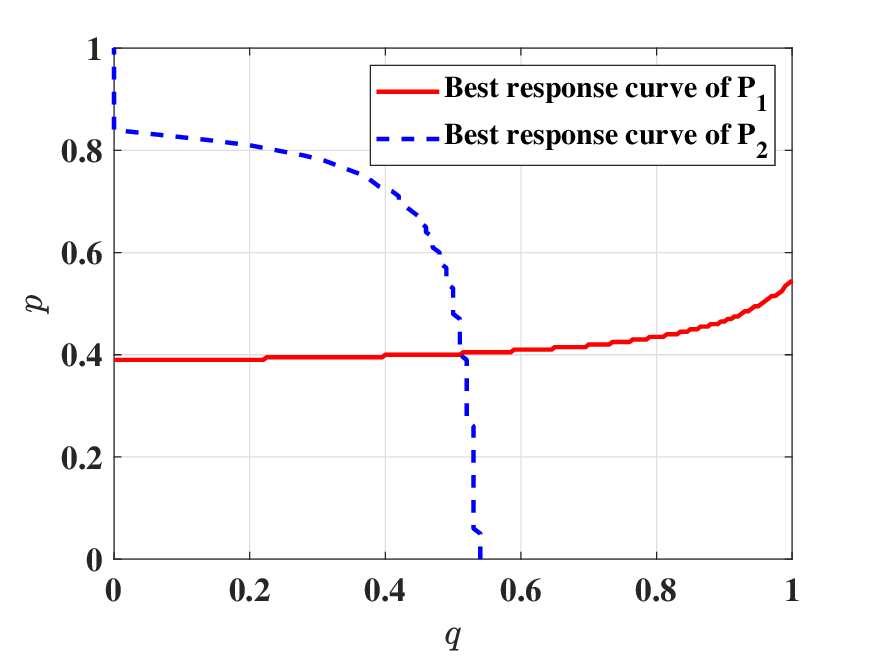}
 \vspace{-0.35cm}
	\caption{\small{Best response plot for each player in \Cref{exp:simpleExample}.
	}}
	\label{Fig:BR_plot}
\end{figure}

\begin{figure}[!b]
	\centering
	\includegraphics[width=0.95\columnwidth ]{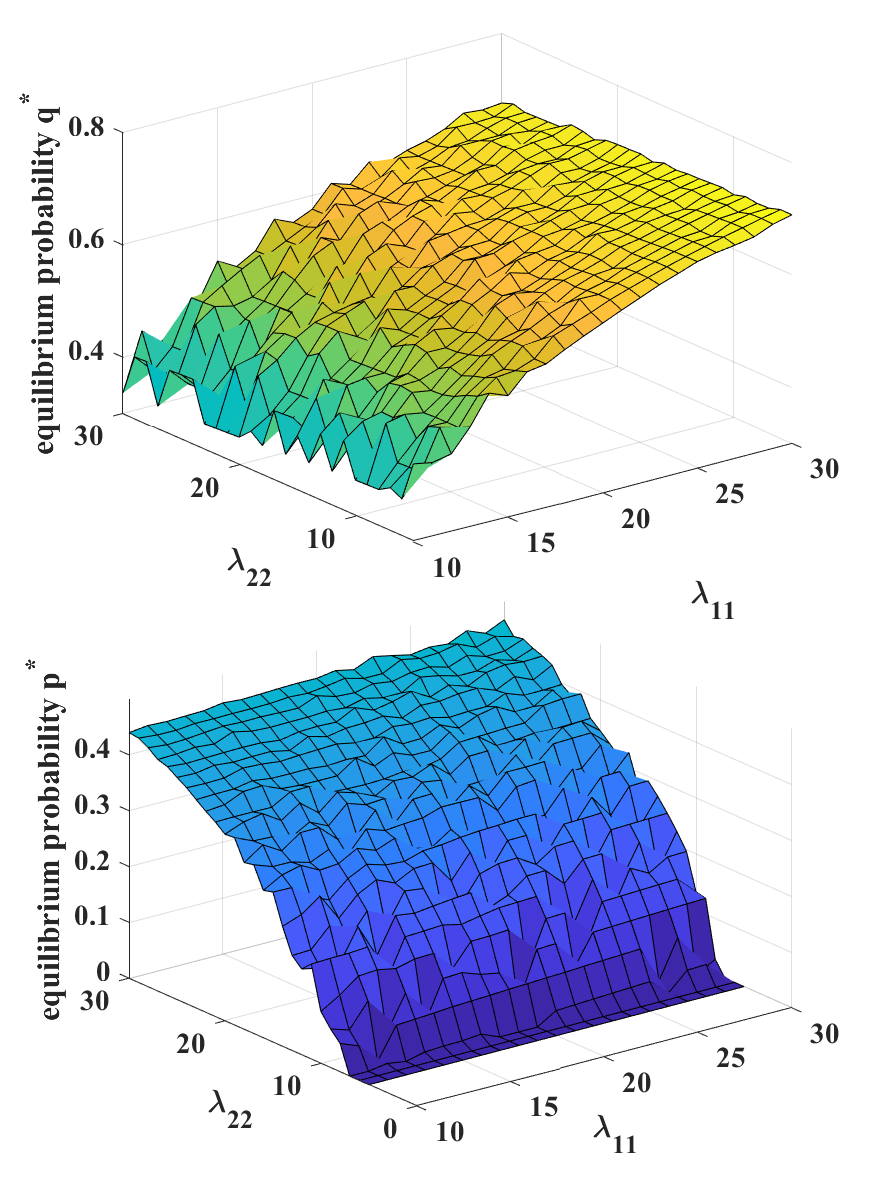}
 \vspace{-0.15cm}
	\caption{\small{Variation of ($p^*,q^*$) versus scheduling costs $\lambda_{11}$ and $\lambda_{22}$ in \Cref{exp:simpleExample}.
	}}
	\label{Fig:eqb_vs_lambda}
 	\vspace{-0.4cm}
\end{figure}

\begin{figure*}
	\centering
	\includegraphics[width=\textwidth ]{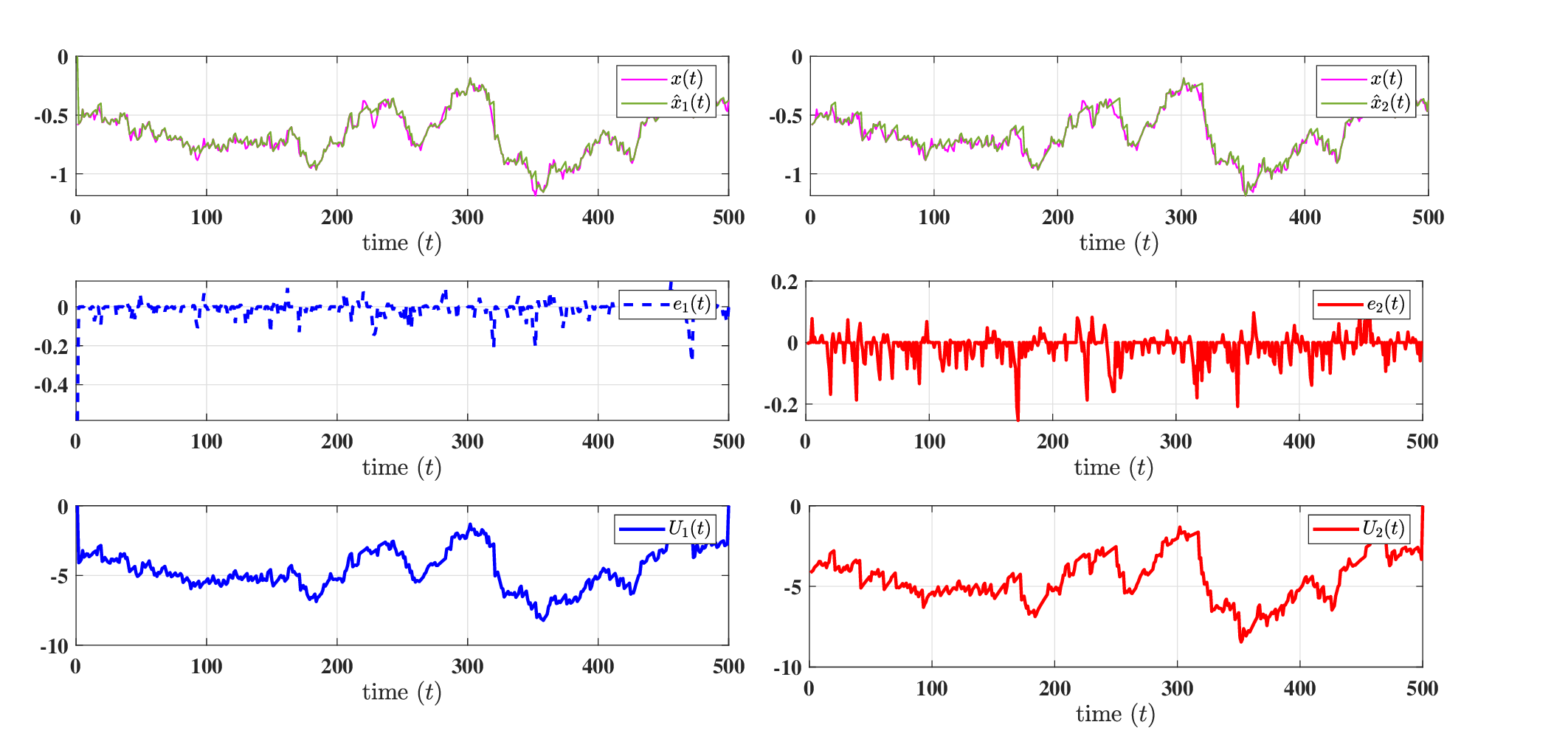}
 \vspace{-0.75cm}
	\caption{\small{Signal trajectories for each player for \Cref{exp:simpleExample}.
	}}
	\label{Fig:real_sim_scalar}
\end{figure*}

\begin{example} \label{exp:simpleExample}
We begin by first considering an unstable scalar system with the parameters: $A=1.5, B_1 = 1, B_2 = 0.5, Q=4, R_1=1, R_2=0.5$ and $G=4$. 
Consequently, we may use \eqref{eq:ARE} to compute the Riccati matrix, which results in $P = 7.1231$. Further, we set $\lambda_{11} = 25, \lambda_{12} = 17, \lambda_{21} = 25$ and $\lambda_{22} = 15$. 
The discretization interval as presented in \eqref{eq:disc} is chosen to be $h=0.01$ and the truncation parameter in Theorem \ref{thm:approximation} is taken to be $t_p=400$. The step sizes and tolerance parameters in Algorithm \ref{alg:Nash} are taken as $\eta = \varepsilon = 0.0001$. 
We then run Algorithm \ref{alg:Nash} for the above system with 10 random initializations of $(p^{(0)},q^{(0)})$. The equilibrium value was obtained to be unique and is given by $p^* = 0.4013$ and $q^* = 0.4934$. 
We also plot the best response curves (i.e., \Cref{alg:exhaustive}) for both players in Fig.~\ref{Fig:BR_plot}, where their intersection point is the Nash equilibrium value. 

For completeness, we also plot the state of the game, their state estimate, error and action trajectories in Fig.~\ref{Fig:real_sim_scalar} for $T = 500$ sec, from which we observe that the state estimates for both player follow the state $x(t)$ quite closely, which implies that the estimation errors for the two players stay bounded for the open-loop unstable system. Finally, in Fig.~\ref{Fig:eqb_vs_lambda}, we also study the variation of the Nash equilibrium pair ($p^*,q^*$) as a function of the scheduling cost coefficients $\lambda_{11}$ and $\lambda_{22}$. From the same, we observe that both $p^*$ and $q^*$ show a general increasing trend (with a slightly irregular peaky surface) with increasing values of $\lambda_{11}$ and $\lambda_{22}$, which means that as the cost to scheduling increases, the players are incentivized to decrease their transmissions, as aligned with intuition.
\end{example}

\begin{example}[Pursuit-Evasion Game] \label{exp:PEG}
    Next, we consider a pursuit-evasion (or an attacker-defender) game, which is a classic example of a two-player game \cite{bacsar1998dynamic,maity2023efficient}.
Let us start by defining the states of the pursuer and the evader as $x_p \in \Re^2$ and $x_e \in \Re^2$, respectively. Consequently, their dynamics are given as:

\begin{align}\label{PE_game}
\begin{split}
    \dx_p(t) &= u_p(t) \dt + G_p \dW(t) \\
    \dx_e(t) &= u_e(t) \dt + G_e \dW(t),
\end{split}
\end{align}
where $u_p(t),u_e(t) \in \Re^2$ denote the control inputs of the pursuer and the evader, respectively. Let us define the state of the game to be $x(t) = x_p(t) - x_e(t)$. The dynamics \eqref{PE_game} then imply that $A = 0_2, B_1 = I_2, B_2 = -I_2$ and $G = G_p-G_e$. Further, we take the objective function of the pursuer to be

\begin{align*}
    J_p & = \limsup_{T \rightarrow \infty} \frac{1}{T}\E \Big[ \int_0^T(\|x(t)\|_{I_2}^2 + \|u_p(t)\|^2_{0.25I_2} \\
    & ~~~~ - \|u_e(t)\|^2_{0.5I_2}) dt + 25 n_p + 17 n_{e}\Big],
\end{align*}

and that of the evader to be
\begin{align*}
    J_e & = \limsup_{T \rightarrow \infty} \frac{1}{T}\E \Big[ \int_0^T(\|x(t)\|_{I_2}^2 + \|u_p(t)\|^2_{0.25I_2} \\
    & ~~~~ - \|u_e(t)\|^2_{0.5I_2}) dt + 25 n_e + 15 n_{p}\Big],
\end{align*}
where $n_p$ and $n_e$ denote the total number of scheduling instants up to time $T$ for the pursuer and the evader, respectively.
One may verify that the above system satisfies \Cref{Assump_1} and subsequently, using \eqref{eq:ARE}, the Riccati matrix is computed to be 
$P = \frac{1}{\sqrt{2}} I_2 
$.
The discretization interval is taken to be $h=0.01$ and the truncation parameter in Theorem \ref{thm:approximation} is taken to be $t_p=400$. The step sizes and tolerance parameters are taken as $\eta=\varepsilon = 0.0001$.
We then ran \Cref{alg:Nash} for the above system with 10 random initializations of $(p^{(0)},q^{(0)})$. The equilibrium value was obtained to be unique for all the different initializations and given by $p^* = 0.8289$ and $q^* = 0.8785$. This implies that the probibilty of scheduling for the players is $1-p^* = 0.1711$ and $1-q^* = 0.1215$. 
Additionally, we also plot the best response curves  (i.e., \Cref{alg:exhaustive}) for both players in Fig.~\ref{Fig:BR_plot_PE}. 
Further, for completeness, the state, estimate, error and action trajectories are plotted for both players in Fig.~\ref{Fig:real_sim_PE} for a horizon of $T=500$ seconds. From the same, we observe that since the scheduling probabilities of both players are quite low, their state estimates do not match very closely with the actual state evolution. The estimation error, however, in both cases is bounded.
\begin{figure}[h]
\vspace{-3mm}
	\centering
    \includegraphics[width=1.07\columnwidth]{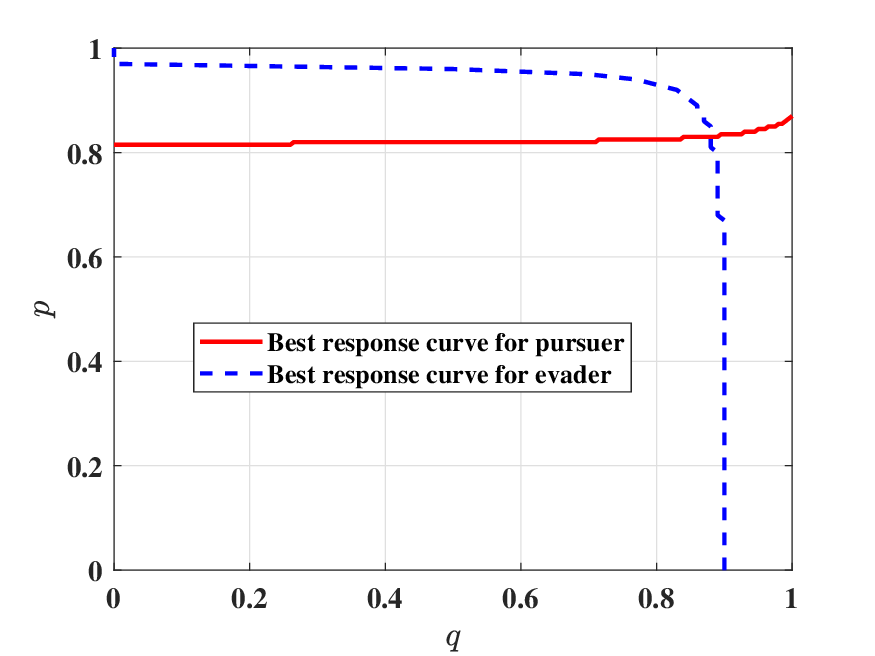}
 \vspace{-0.35cm}
	\caption{\small{Best response plot for each player in \Cref{exp:PEG}.
	}}
	\label{Fig:BR_plot_PE}
 	\vspace{-0.6cm}
\end{figure}
\end{example}

\begin{figure*}[h]
	\centering
	\includegraphics[width=1.07\textwidth ]{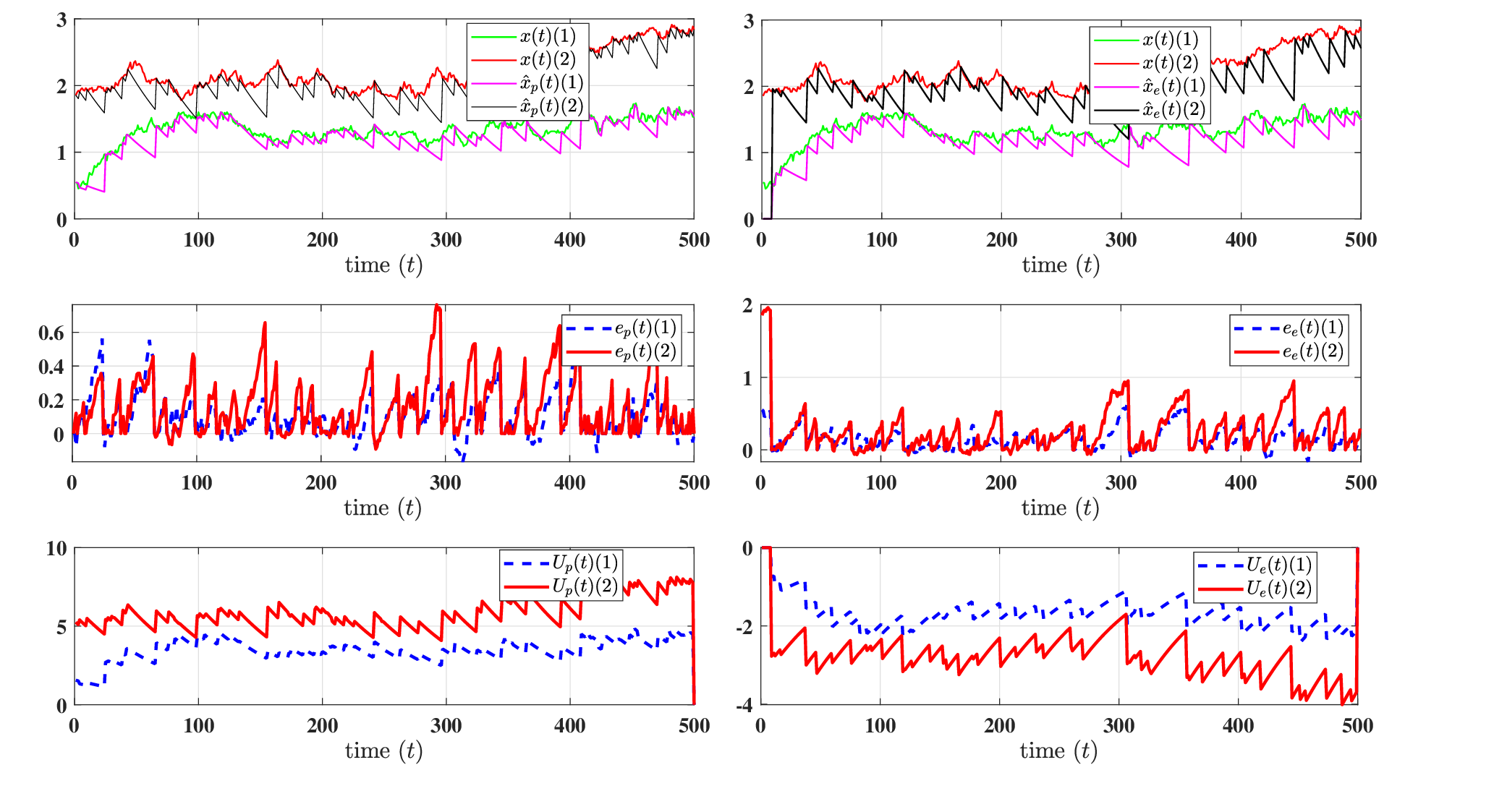}
 \vspace{-0.75cm}
	\caption{\small{Signal trajectories for each player for \Cref{exp:PEG}.
	}}
	\label{Fig:real_sim_PE}
 	\vspace{-0.1cm}
\end{figure*}

\section{Conclusions \& Discussion}\label{sec:conc_disc}
In this article, we have studied the communication-control tradeoff for a class of two-player nonzero-sum differential games. 
Specifically, each player's scheduler acts as an active decision maker to only intermittently communicate the state of the game to its remotely located controller due to an associated cost with each communication. 
Subsequently, each player's controller relies on its causal information to control the state. Within this setup, we have first provided (Nash) equilibrium controller policies for the players based on the conditional estimate of the state. Consequently, we have reformulated the scheduler optimization problems for the two players into an auxiliary problem involving the steady state solution to a generalized Sylvester equation. Since the latter is the sum of infinite powers of a polynomial, we have first provided an exhaustive search algorithm to plot the best response curves for both players using a truncated polynomial. 
Consequently, we have also provided an iterative algorithm to numerically compute a Nash equilibrium policy for the schedulers of the players. Finally, we have provided a rigorous performance evaluation to assess the merits of our approach on two numerical examples. From the results, the pair of Nash equilibrium values for the scheduling policies appears to be unique for multiple different intitializations of the system parameters.

There are several research directions that emanate from our work. Firstly, we note that our information sets of the controllers are oblivious to the information contained within the \textit{no-communication events}.
Specifically, each instant when information is not relayed from scheduler to the controller of a player can additionally inform the conditional estimate of the state at the controller. 
This, however, leads to the estimator becoming nonlinear, and consequent failure of the separation to hold between the design of the controller and the sensor policy for each player. 
A single player optimal control problem taking into account the same has been considered in \cite{soleymani2021value,soleymani2022value}. 
It would be interesting to investigate the case of a two-player Nash game under such a setting.

Another plausible direction worth investigating is to look for equilibrium scheduling policies within a wider class of policies. Here, we used a parametrized class of open-loop policies while it would be interesting to look at a class of closed-loop or feedback scheduling policies.

\bibliography{references,maity}
\bibliographystyle{IEEEtran}

\section*{Appendix I \\ Proof of Theorem \ref{thm:Nash_control}}\label{App:proof1}

To prove that \eqref{eq:eqb_control}-\eqref{eq:eqb_xhat2} provide an equilibrium pair, we show that $\mu^*_{C2}$ is the \textit{best response} to $\mu^*_{C1}$ and vice versa. 
To that end, let us start by fixing $\Pa$'s policy to $\mu^*_{C1}$ and derive $\Pb$'s strategy and show that the corresponding maximizing strategy is indeed $\mu^*_{C2}$.

To start our analysis, let us begin by recalling our definition $\tilde A:= A - P^{-1}\Lambda_1 + P^{-1} \Lambda_2$. Then, we can rewrite $\Pa$'s estimator  \eqref{eq:eqb_xhat1} as
\begin{align*}
    \dot{\hat x}^*_1(t) & = \tilde{A} {\hat x}_1^*(t), ~ t \in [0,\infty) \setminus {\tt T}_1 \\
            \hat x_1^*(t_\ell) & = x(t_\ell), ~~t_\ell \in {\tt T}_1.
\end{align*}

Now, define $z \triangleq [x^\top ~ \hat x_1^\top]^\top$ as the extended state vector. 
After substituting $\Pa$'s controller \eqref{eq:eqb_control} in the dynamics \eqref{eq:dyn}  the dynamics of $z$ follows:
\begin{align}\label{eq:z_dyn}
    \dz \!=\! \Bigg( \underbrace{\begin{bmatrix}
        A & -P^{-1} \Lambda_1 \\
        0 & \tilde A
    \end{bmatrix}}_{=: F} z + \underbrace{\begin{bmatrix}
        B_2 \\
        0
    \end{bmatrix}}_{=: \tilde B_2} u_2\Bigg) \dt \!+\! \underbrace{\begin{bmatrix}
        G \\
        0
    \end{bmatrix}}_{=: \tilde G} \dW(t),
\end{align}
 for all $t \in [0,\infty) \setminus {\tt T}_1$ and $z(t_\ell) = [x(t_\ell)^\top ~ x(t_\ell)^\top]^\top$ for all $t \in {\tt T}_1$. Consequently, by substituting the control policy of $\Pa$ in \eqref{eq:pre_control_cost}, one can also obtain the objective function to be minimized by $\Pb$ as:

 \begin{align}\label{eq:z_obj}
     J = \limsup_{T \rightarrow \infty} \frac{1}{T} \E \Big[ \int_0^T \|z\|^2_{\tilde Q} + \|u_2\|^2_{R_2} \dt \Big]
 \end{align}
where we define $\tilde Q:= {\tt Diag}[-Q, -\Lambda_1]$.
Thus, $\Pb$ must solve the optimal control problem of minimizing  \eqref{eq:z_obj} subject to the dynamics \eqref{eq:z_dyn}. This is a  linear quadratic Gaussian (LQG) optimal control problem with indefinite cost, the solution to which is given as:
\begin{align}\label{eq:control_2}
            u_2^*(t)  = - R_2^{-1} \tilde B_2^\top \tilde P \hat{z}(t),
        \end{align}
where $\hat z(t) = \E[z(t) \mid {\tt I_{C_2}}(t)]$ denotes $\Pb$'s  estimate of the state $z(t)$ and the Riccati matrix $\tilde P$ satisfies the ARE

\begin{align}\label{eq:ARE_Ptilde}
    \tilde P F + F^\top \tilde P + \tilde Q - \tilde P \tilde B_2 R_2^{-1} \tilde B_2^\top \tilde P = 0.
\end{align}

At this point, two items are to be noted: (i) the policy in \eqref{eq:control_2} is not in the form of $\mu^*_{C2}$ in \eqref{eq:eqb_control}, and (ii) the LQG problem in \eqref{eq:z_obj} is non-standard since $\tilde Q$ is negative semidefinite (instead of being \textit{positive} semidefinite), and therefore, the existence of $\tilde{P}$ must be verified. 
The rest of the proof will focus on these two aspects and we start with item (i) in the following.

Note that $\hat z(t) = \E[z(t) \mid {\tt I_{C_2}}(t)]$ follows the dynamics
\begin{align}\label{eq:zhat_dyn}
    \dot{\hat z}(t) = F \hat z(t) + \tilde B_2 u_2(t), ~~ \forall t \in [0,\infty) \setminus {\tt T}_2,
\end{align}
and $\hat z(t_\ell') = [x(t_\ell') ~ \hat x_{12}(t_\ell')]^\top$ for $t_\ell' \in {\tt T}_2$ since $\Pb$ observes $x$ perfectly at the scheduling instances in ${\tt T}_2$. Here, we have defined $\hat x_{12}(t): = \E[\hat x_1(t) \mid {\tt I_{C_2}}(t)]$. 

Next, we will simplify the expression \eqref{eq:control_2}. 
To that end, note that $R_2^{-1} \tilde B_2^\top \tilde P \hat{z}(t) = R_2^{-1} B_2^\top [I~~ 0] \tilde P \hat{z}(t)$.
    Now, let us define
\begin{align} \label{eq:xi}
    [I~~0] \tilde P \hat z(t) = -P \xi(t),
\end{align}
for some signal $\xi(t)$ to be defined soon. 
Upon differentiating both sides of \eqref{eq:xi} yields

\begin{align} \label{eq:xi_dot}
    [I~~0] \tilde P \dot{\hat z}(t) = -P \dot{\xi}(t).
\end{align}
Subsequently, upon substituting the $\hat z$
dynamics from \eqref{eq:zhat_dyn} and the controller \eqref{eq:control_2} in \eqref{eq:xi_dot}, we arrive at

\begin{align*}
    -P\dot{\xi}(t) = [I~0] (\tilde P F - \tilde P \tilde B_2 R_2^{-1} \tilde B_2^\top \tilde P) \hat z(t).
\end{align*}
Using \eqref{eq:ARE_Ptilde} to further simplify the last equation yields

\begin{align}\label{eqn:third_diff_eqn}
    -P \dot{\xi}(t) &= -[I~0]F^\top \tilde P \hat z(t) - [I~0]\tilde Q \hat z(t) \nonumber \\
    & = -[A^\top ~0] \tilde P \hat{z}(t) + Q \hat{x}_2 \nonumber \\
    & = A^\top P \xi(t) + Q \hat{x}_2 \nonumber \\
    & = A^\top P (\xi(t) - \hat x_2(t)) - P \tilde A \hat x_2(t),
\end{align}
where we have used the GARE \eqref{eq:ARE} to substitute $Q$ and we define $\hat x_2(t) = \bbE[x(t) \mid {\tt I_{C_2}}(t)]$.

Note that, with help of \eqref{eq:xi}, we may now rewrite 
\begin{align*}
    u_2^*(t) =  R_2^{-1}B_2^\top P \xi(t)
\end{align*}
At this point, it only remains to show that $\xi(t) = \hat x_2(t)$ for all $t \in [0,\infty)$ to conclude that $\mu^*_{C2}$ in \eqref{eq:eqb_control} is the best response to $\mu^*_{C1}$.
Indeed, if $\xi\equiv \hat{x}_2$, then \eqref{eqn:third_diff_eqn} yields $\dot{\hat{x}}_2(t) = \tilde{A}\hat{x}_2(t)$, which coincides with the dynamics \eqref{eq:eqb_xhat2}.
Now, let us formally show that $\xi\equiv \hat{x}_2$ with $\hat x_2$ following the \textit{jump} differential equation \eqref{eq:eqb_xhat2}. 




To that end, let us start by substituting \eqref{eq:control_2} into \eqref{eq:zhat_dyn}:
\begin{subequations} \label{eqn:two_diff_eq}
\begin{align}
    \dot{\hat x}_2(t) & = A \hat x_2(t) -P^{-1}\Lambda_1 \hat x_{12}(t) + P^{-1} \Lambda_2 \xi(t)  \\
    \dot{\hat x}_{12}(t) &= \tilde A {\hat x}_{12}(t),
\end{align}
\end{subequations}
where recall that $\hat{z}(t) =  [\E[x(t) \mid {\tt I_{C_2}}(t)]^{\!\top} ~  \E[ \hat x_1(t) \mid {\tt I_{C_2}}(t)]^{\!\top}]^\top $ $ = [ \hat{x}_2(t)^\top ~ \hat{x}_{12}(t)^\top ]^\top$.
At the communication instances $t_{\ell} \in {\tt T}_2$, we define 
$\hat{x}_2(t_\ell) = x(t_\ell)$
and $\hat{x}_{12}(t_\ell) = x(t_\ell)$. 
One may then verify that $\hat{x}_2(t) = \hat{x}_{12}(t)$ for all $t$, which we leave as an exercise to the reader.\footnote{A similar condition was also derived in \cite[equation (A.14)]{maity2017linear} that adopted a significantly different approach involving Gateaux differentials. 
Essentially, what $\hat{x}_2(t) = \hat{x}_{12}(t)$ implies is that $\Pb$'s estimate of the state and $\Pb$'s estimate of $\Pa$'s estimate of the state coincide.}
Now, $\hat{x}_2(t) = \hat{x}_{12}(t)$ in conjunction with \eqref{eqn:two_diff_eq} results in $\hat{x}_2(t) = \hat{x}_{12}(t) = \xi(t)$ for all $t$. 


Next, let us show that  $\tilde P$ is well defined. Let us denote
\begin{align*}
    \tilde P = \begin{bmatrix}
        \tilde P_{11} & \tilde P_{12} \\
        \tilde P_{12} & \tilde P_{22}
    \end{bmatrix}.
\end{align*}
Then, multiplying \eqref{eq:ARE_Ptilde} on the left by $[I~0]$ and on the right by $[I~0]^\top$, we obtain the ARE
\begin{align*}
    \tilde P_{11} A + A^\top \tilde P_{11} -Q - \tilde P_{11} B_2 R_2^{-1} B_2^\top \tilde P_{11} = 0.
\end{align*}
Then, using the hypotheses of Theorem \ref{thm:Nash_control}, the solution to the above ARE exists, which also implies that the matrix $A_P:= A - \tilde  B_2 R_2^{-1} B_2^\top \tilde P_{11}$ is Hurwitz. 

Next, we prove that the existence of $\tilde P_{11}$ implies existence of $\tilde P_{12}$. Let us begin by defining 
\begin{align*}
    \tilde P_a := \tilde P_{11} + \tilde P_{12}.
\end{align*}
Then, by multiplying \eqref{eq:ARE_Ptilde} on the left by $[I~0]$ and on the right by $[I~I]^\top$, we arrive at the following equation:
\begin{align*}
     A^\top & \tilde P_a + \tilde P_a A - Q - \tilde P_a B_1R_1^{-1}B_1^\top P + \tilde P_{12} B_2R_2^{-1}B_2^\top P \\
    & \hspace{2.5cm} - \tilde P_{11} B_2R_2^{-1} B_2^\top \tilde P_a = 0.
\end{align*}
Adding this equation to \eqref{eq:ARE} yields
\begin{align*}
    A^\top & (\tilde P_a + P) + (\tilde P_a + P) A - (\tilde P_a + P) B_1R_1^{-1}B_1^\top P \\
    & + (\tilde P_a + P) B_2R_2^{-1}B_2^\top P + \tilde P_{12} B_2R_2^{-1} B_2^\top (\tilde P_a + P) \\
    & - \tilde P_a B_2R_2^{-1} B_2^\top (\tilde P_a + P) = 0,
\end{align*}
which, upon combining terms, can be re-written as
\begin{align*}
    (\tilde P_a + P) \tilde A + A_P^\top  (\tilde P_a + P) = 0. 
\end{align*}
This implies that
\begin{align*}
    e^{\tilde At} (\tilde P_a + P) e^{A_P^\top t} = \text{constant}.
\end{align*}
Since both $A_P$ and $\tilde A$ are Hurwitz, we conclude that $\tilde P_a + P = 0$ or equivalently, $\tilde P_{12} = -P - \tilde P_{11}$. 
This, together with the existence of the matrix $P$ under Assumption \ref{Assump_1}, and the existence of $\tilde P_{11}$  as discussed before, implies the existence of $\tilde P_{12}$. 
Finally, by multiplying \eqref{eq:ARE_Ptilde} by $[0~I]$ on the left and by $[0~I]^\top$ on the right, we obtain the ARE
\begin{align*}
    \tilde P_{22} \tilde A & + \tilde A^\top \tilde P_{22} - \tilde P_{12} P^{-1}\Lambda_1 - \Lambda_1^\top P^{-1} \tilde P_{12} - \Lambda_1 \\
    & - \tilde P_{12} B_2 R_2^{-1} B_2^\top \tilde P_{12} = 0.
\end{align*}
Then, since $\tilde P_{12}$ exists and $\tilde A$ is Hurwitz (from Theorem \eqref{thm:Nash_control}), we conclude that $\tilde P_{22}$ exists. This finally implies that $\tilde P$ exists, and completes our analysis for the policy of player 2.

Let us now turn to showing that the best response policy of $\Pa$ is given by $\mu^*_{C1}$ when $\Pb$ follows the policy $\mu^*_{C2}$. This can be proven by repeating exactly the same steps as above, and hence, we do not detail them here. Rather, we only show that we do not require any additional existence condition for $\Pa$ analogous to the ARE \eqref{eq:ARE_Ptilde} for $\Pb$. Let us consider the following ARE which would appear analogous to \eqref{eq:ARE_Ptilde} when computing the optimal control policy for $\Pa$ under $\mu^*_{C2}$:
\begin{align}\label{eq:ARE_Ptilde_P1}
    \tilde M F + F^\top \tilde M + \tilde{\tilde Q} - \tilde M \tilde B_2 R_2^{-1} \tilde B_2^\top \tilde M = 0,
\end{align}
with $\tilde{\tilde Q} = {\tt Diag}[Q~-\Lambda_2]$ and $\tilde M$ defined as
\begin{align*}
    \tilde M = \begin{bmatrix}
        \tilde M_{11} & \tilde M_{12} \\
        \tilde M_{12} & \tilde M_{22}
    \end{bmatrix}.
\end{align*}
Then, multiplying \eqref{eq:ARE_Ptilde_P1} on the left by $[I~0]$ and on the right by $[I~0]^\top$, one can arrive at the equation:
\begin{align*}
    \tilde M_{11} A + A^\top \tilde M_{11} + Q - \tilde M_{11} B_1R_1^{-1} B_1^\top \tilde M_{11} = 0.
\end{align*}
The existence of $\tilde M_{11}$ then follows under Assumption \ref{Assump_1}. This then also implies that the matrix $A_M = A - B_1R_1^{-1}B_1^\top \tilde M_{11}$ is Hurwitz. 

Next, multiplying \eqref{eq:ARE_Ptilde_P1} on the left by $[I~0]$ and on the right by $[I~I]^\top$, and by letting $X := \tilde M_{11} + \tilde M_{12} -P$, we arrive at the following equation: 
\begin{align*}
    & XA + A^\top X + X B_2R_2^{-1} B_2^\top P - X B_1 R_1^{-1} B_1^\top P \\
    & - \tilde M_{11} B_1 R_1^{-1} B_1^\top X = 0.
\end{align*}
Upon combining terms, the above equation can be further written as:
\begin{align*}
    X \tilde A + A_M X = 0
\end{align*}
which leads to the expression
    $e^{A_M t} X e^{\tilde A t} = \text{constant}$.
Then, since both $\tilde A$ and $A_M$ are Hurwitz matrices, we conclude that $X=0$ or equivalently, $\tilde M_{12} = P - \tilde M_{11}$. This then implies the existence of $\tilde M_{12}$. The existence of $\tilde M_{22}$ can be proven again by multiplying both sides of \eqref{eq:ARE_Ptilde_P1} by $[0~I]$ and its transpose on the left, and right, respectively. This entails the existence of $\tilde M$ and completes the proof of the theorem.

\section*{Appendix II \\ Proof of Proposition \ref{Prop:cov_dynamics}}
Let us begin by defining the state transition matrix of 
\begin{align*}
    \Phi_{\bar A} (t) = e^{\bar A t} =: \begin{bmatrix}
        \Phi_{11}(t) & \Phi_{12}(t) \\
        \Phi_{21}(t) & \Phi_{22}(t)
    \end{bmatrix},
\end{align*}
where recall that (see the unnumbered equation after \eqref{eq:err_combined})
\begin{align*}
    \bar A = \begin{bmatrix}
        A + P^{-1} \Lambda_2 & -P^{-1} \Lambda_2 \\
        P^{-1} \Lambda_1 & A - P^{-1} \Lambda_1
    \end{bmatrix} .
\end{align*}
Then, the estimation error dynamics in \eqref{eq:err_combined_disc} can be re-written as
\begin{align}
    e_{1,k+1} & = (1-\gamma_{1,k+1}) \Big[ \Phi_{11}(h) e_{1,k} + \Phi_{12}(h) e_{2,k} \nonumber \\
    & \qquad + \int_0^h [\Phi_{11}(h-s) + \Phi_{12}(h-s)]G \dW(s)\Big] \label{eq:error_dyn_discretized1} \\
    e_{2,k+1} & = (1-\gamma_{2,k+1}) \Big[ \Phi_{21}(h) e_{1,k} + \Phi_{22}(h) e_{2,k} \nonumber \\
    & \qquad + \int_0^h [\Phi_{21}(h-s) + \Phi_{22}(h-s)]G \dW(s)] \label{eq:error_dyn_discretized2}
\end{align}

Then, using \eqref{eq:error_dyn_discretized1}-\eqref{eq:error_dyn_discretized2}, one can show that $\Sigma_{11,k}$ satisfies the following recursive equation: 

\begin{align*}
    & \Sigma_{11,k+1} = \mathbb E [e_{1,k+1} e_{1,k+1}^\top] \\
    & \!= \! \mathbb E[(1-\gamma_{1,k+1})^2] \Big[\Phi_{11}(h) \Sigma_{11,k} \Phi_{11}^\top(h) \!+\! \Phi_{11}(h) \Sigma_{12,k} \Phi_{12}^\top(h) \\
    & ~~~ + \Phi_{12}(h) \Sigma_{12,k}^\top \Phi_{11}^\top(h) + \Phi_{12} \Sigma_{22,k} \Phi_{12}^\top(h) + \tilde G_1(h)\Big] \\
    & = p \Big[\Phi_{11}(h) \Sigma_{11,k} \Phi_{11}^\top(h) + \Phi_{11}(h) \Sigma_{12,k} \Phi_{12}^\top(h) \\
    & ~~~ + \Phi_{12}(h) \Sigma_{12,k}^\top \Phi_{11}^\top(h) + \Phi_{12} \Sigma_{22,k} \Phi_{12}^\top(h) + \tilde G_1(h)\Big] \\
\end{align*}
where we define $\tilde G_1(h)$ as
\begin{align*}
    \tilde G_1(h) = \int_0^h (\Phi_{11}(h-s) & + \Phi_{12}(h-s)) GG^\top \\
    & \times (\Phi_{11}(h-s) + \Phi_{12}(h-s))^\top \ds.
\end{align*}

Similarly, one can also obtain recursions for $\Sigma_{12,k}$ and $\Sigma_{22,k}$ as follows:

\begin{align*}
    \Sigma_{12,k+1} & = pq \Big[ \Phi_{11}(h) \Sigma_{11,k} \Phi_{21}^\top(h) + \Phi_{11}(h) \Sigma_{12,k} \Phi_{22}^\top(h) \\
    & \hspace{-5mm} + \Phi_{12}(h) \Sigma_{12,k}^\top \Phi_{21}^\top(h) + \Phi_{12}(h) \Sigma_{22,k} \Phi_{22}^\top(h) + \tilde G_2(h) \Big] \\
    \Sigma_{22,k+1} & = q (\Phi_{21} (h)\Sigma_{11,k} \Phi_{21}^\top(h) + \Phi_{21}(h) \Sigma_{12,k} \Phi_{22}^\top(h) \\
    & \hspace{-5mm} + \Phi_{22}(h) \Sigma_{12,k}^\top \Phi_{21}^\top(h) + \Phi_{22}(h) \Sigma_{22,k} \Phi_{22}^\top(h) + \tilde G_3(h)\Big],
\end{align*}
where we define $\tilde G_2(h)$ and $\tilde G_3(h)$ as
\begin{align*}
    \tilde G_2(h) = \int_0^h (\Phi_{11}(h-s) & + \Phi_{12}(h-s)) GG^\top \\
    & \times (\Phi_{21}(h-s) + \Phi_{22}(h-s))^\top \ds \\
    \tilde G_3(h) = \int_0^h (\Phi_{21}(h-s) & + \Phi_{22}(h-s)) GG^\top \\
    & \times (\Phi_{21}(h-s) + \Phi_{22}(h-s))^\top \ds.
\end{align*}
This then completes the proof of the Proposition.

\section*{Appendix III \\ Gradient Computation}
Here, we detail the procedure to compute the derivative of the objective with respect to the decision variable $p$; the one with respect to $q$ can be carried out analogously. Let us begin by differentiating \eqref{eq:ss_cov} on both sides to arrive at the partial differential equation
\begin{align*}
    \frac{\partial \Sigma_\infty}{\partial p} & \!=\! \!\!\!\sum_{i=1,2,3} \!\! \frac{\partial A_i(p,q)}{\partial p} \Sigma_\infty A_i(p,q)^\top \! \!+ \! A_i(p,q) \frac{\partial \Sigma_\infty}{\partial p} A_i(p,q)^\top \nonumber \\
    & \hspace{2cm} + A_i(p,q) \Sigma_\infty \frac{\partial A_i(p,q)^\top}{\partial p}  + \frac{\partial G(p,q)}{\partial p}.\\
\end{align*}
Define $\Sigma_\infty^{(p)}:= \frac{\partial \Sigma_\infty}{\partial p}$ and $\Sigma_\infty^{(q)}:= \frac{\partial \Sigma_\infty}{\partial q}$. Then, we have that
\begin{subequations}\label{eq:grad_comp}
    \begin{align}
    \Sigma_\infty^{(p)} & = \sum_{i=1,2,3} A_i(p,q) \Sigma_\infty^{(p)} A_i(p,q)^\top + \mathcal{K} \\
    \mathcal{K} & := \sum_{i=1,2,3}\frac{\partial A_i(p,q)}{\partial p} \Sigma_\infty A_i(p,q)^\top \nonumber \\
    & \hspace{0.5cm}+ A_i(p,q) \Sigma_\infty \frac{\partial A_i(p,q)^\top}{\partial p}  + \frac{\partial G(p,q)}{\partial p}.
\end{align}
\end{subequations}
Now, if the sufficient condition on $\rho(\cdot)$ in Theorem \ref{thm:approximation} is satisfied, it guarantees the existence of the solution to \eqref{eq:grad_comp}.
Then, an iterative solution  to the Sylvester type equation \eqref{eq:grad_comp} can be obtained using the first result in Theorem \ref{thm:approximation} to compute an approximate value of $\Sigma_\infty'$, which can subsequently be used to compute the gradients as
\begin{align*}
    \frac{\partial J}{\partial p} = \tr\Big( \Lambda \Sigma_\infty^{(p)} \Big) + \lambda_{11} \\
    \frac{\partial J}{\partial q} = \tr\Big( \Lambda \Sigma_\infty^{(q)} \Big) - \lambda_{22}.
\end{align*}
This completes the gradient computation.

\begin{IEEEbiography}[{\includegraphics[width=1in,height=1.25in,clip,keepaspectratio]{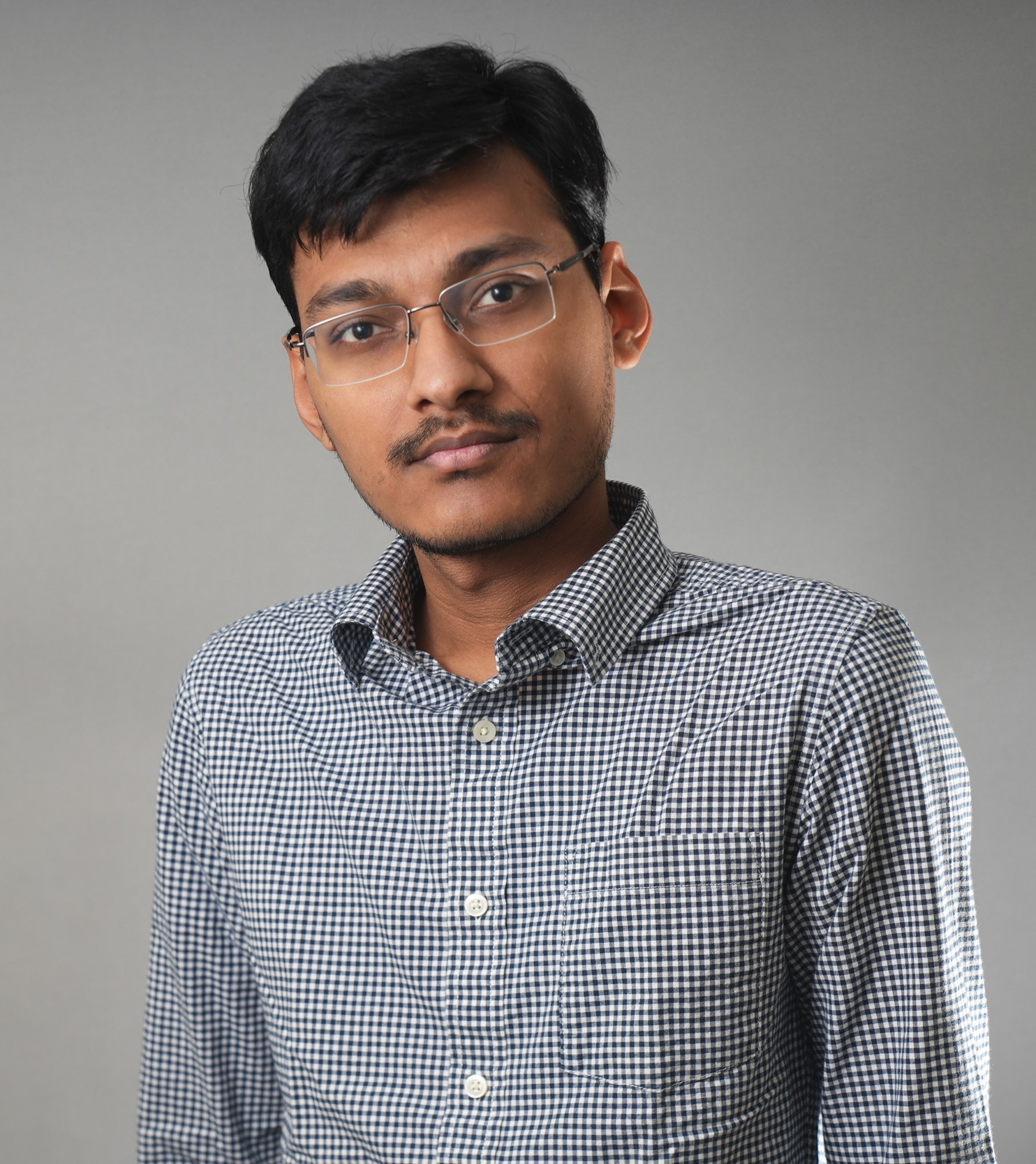}}]{Shubham Aggarwal} (Student Member,~IEEE) received his B.Tech. and M.Tech. in Electrical Engineering from the Indian Institute of Technology (BHU), Varanasi, India. He also completed his M.S. in Mathematics at the University of Illinois Urbana Champaign (UIUC), IL, USA, where he is currently pursuing his Ph.D. in the Coordinated Science Laboratory and the Department of Mechanical Science \& Engineering. His current research interests include wireless communication, optimal control theory, large multi-agent systems, and reinforcement learning.
\end{IEEEbiography}

\begin{IEEEbiography}[{\includegraphics[width=1in,height=1.25in,clip,keepaspectratio]{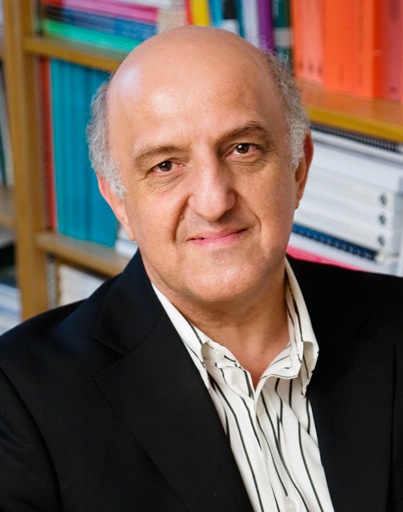}}]{Tamer Ba\c{s}ar} (S’71-M’73-SM’79-F’83-LF’13) has been with the University of Illinois Urbana-Champaign since 1981, where he is currently Swanlund Endowed Chair Emeritus and Center for Advanced Study (CAS) Professor Emeritus of Electrical and Computer Engineering, with also affiliations with the Coordinated Science Laboratory, Information Trust Institute, and Mechanical Science and Engineering. At Illinois, he has also served as Director of CAS (2014-2020), Interim Dean of Engineering (2018), and Interim Director of the Beckman Institute (2008-2010). He received B.S.E.E. from Robert College, Istanbul, and M.S., M.Phil., and Ph.D. from Yale University, from which he received in 2021 the Wilbur Cross Medal. He is a member of the US National Academy of Engineering and a Fellow of the American Academy of Arts and Sciences, as well as Fellow of IEEE, IFAC, and SIAM. He has served as president of IEEE CSS (Control Systems Society), ISDG (International Society of Dynamic Games), and AACC (American Automatic Control Council). He has received several awards and recognitions over the years, including the highest awards of IEEE CSS, IFAC, AACC, and ISDG, the IEEE Control Systems Award, and a number of international honorary doctorates and professorships. He has over 1000 publications in systems, control, communications, optimization, networks, and dynamic games, including books on non-cooperative dynamic game theory, robust control, network security, wireless and communication networks, and stochastic networked control. He was the Editor-in-Chief of Automatica between 2004 and 2014, and is currently editor of several book series. His current research interests include stochastic teams, games, and networks; risk-sensitive estimation and control; mean-field game theory; multi-agent systems and learning; data-driven distributed optimization; epidemics modeling and control over networks; strategic information transmission, spread of disinformation, and deception; security and trust; energy systems; and cyber-physical systems.

\end{IEEEbiography}

\begin{IEEEbiography}[{\includegraphics[width=1in,height=1in,clip,keepaspectratio]{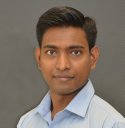}}]{Dipankar Maity} (Senior Member, IEEE) received the B.E. degree in electronics and telecommunication engineering from Jadavpur University, India, in 2013, and the Ph.D. degree in electrical and computer engineering from the University of Maryland, College Park, MD, USA in 2018.
He is an Assistant Professor with the Department of Electrical and Computer Engineering, University of North Carolina at Charlotte, Charlotte, NC, USA. He was a Postdoctoral Fellow with the Georgia Institute of Technology, Atlanta, GA, USA. During his Ph.D., he was a Visiting Scholar at the Technical University of Munich (TUM) and at the KTH Royal Institute of Technology, Stockholm, Sweden. His research interests include temporal logic-based controller synthesis, control under communication constraints, intermittent-feedback control, multi-agent systems, stochastic games, and integration of these ideas in the context of cyber–physical systems.

\end{IEEEbiography}

\end{document}

%% file: arxivmacros.tex
\newcommand{\calM}{\mathcal{M}}

\renewcommand{\Re}{\mathbb{R}}

\newcommand{\bbE}{\mathbb E}

\newcommand{\E}{\mathbb E}

\newcommand{\N}{\mathcal N}

\newcommand{\R}{\mathbb{R}}

\newtheorem{theorem*}{Theorem}[section]

\newtheorem{proposition}{Proposition}[section]

\newtheorem{example}{Example}


%% file: AA_doc.bbl
\begin{thebibliography}{10}
\providecommand{\url}[1]{#1}
\csname url@samestyle\endcsname
\providecommand{\newblock}{\relax}
\providecommand{\bibinfo}[2]{#2}
\providecommand{\BIBentrySTDinterwordspacing}{\spaceskip=0pt\relax}
\providecommand{\BIBentryALTinterwordstretchfactor}{4}
\providecommand{\BIBentryALTinterwordspacing}{\spaceskip=\fontdimen2\font plus
\BIBentryALTinterwordstretchfactor\fontdimen3\font minus
  \fontdimen4\font\relax}
\providecommand{\BIBforeignlanguage}[2]{{%
\expandafter\ifx\csname l@#1\endcsname\relax
\typeout{** WARNING: IEEEtran.bst: No hyphenation pattern has been}%
\typeout{** loaded for the language `#1'. Using the pattern for}%
\typeout{** the default language instead.}%
\else
\language=\csname l@#1\endcsname
\fi
#2}}
\providecommand{\BIBdecl}{\relax}
\BIBdecl

\bibitem{gao2024attention}
A.~Gao, Q.~Wang, Y.~Wang, C.~Du, Y.~Hu, W.~Liang, and S.~X. Ng, ``Attention
  enhanced multi-agent reinforcement learning for cooperative spectrum sensing
  in cognitive radio networks,'' \emph{IEEE Transactions on Vehicular
  Technology}, 2024.

\bibitem{meng2023gnn}
C.~Meng, M.~Tang, M.~Setayesh, and V.~W. Wong, ``{GNN}-based neighbor selection
  and resource allocation for decentralized federated learning,'' in \emph{IEEE
  Global Communications Conference}, 2023, pp. 1223--1228.

\bibitem{shishika2020review}
D.~Shishika and V.~Kumar, ``A review of multi agent perimeter defense games,''
  in \emph{Decision and Game Theory for Security: 11th International
  Conference, GameSec 2020, College Park, MD, USA, October 28--30, 2020,
  Proceedings 11}.\hskip 1em plus 0.5em minus 0.4em\relax Springer, 2020, pp.
  472--485.

\bibitem{alpcan2010network}
T.~Alpcan and T.~Ba{\c{s}}ar, \emph{Network Security: A Decision and
  Game-Theoretic Approach}.\hskip 1em plus 0.5em minus 0.4em\relax Cambridge
  University Press, 2010.

\bibitem{weintraub2008markov}
G.~Y. Weintraub, C.~L. Benkard, and B.~Van~Roy, ``Markov perfect industry
  dynamics with many firms,'' \emph{Econometrica}, vol.~76, no.~6, pp.
  1375--1411, 2008.

\bibitem{hatanaka2015passivity}
T.~Hatanaka, N.~Chopra, M.~Fujita, and M.~W. Spong, \emph{Passivity-based
  {C}ontrol and {E}stimation in {N}etworked {R}obotics}.\hskip 1em plus 0.5em
  minus 0.4em\relax Springer, 2015.

\bibitem{chi2021game}
C.~Chi, Y.~Wang, X.~Tong, M.~Siddula, and Z.~Cai, ``Game theory in internet of
  things: A survey,'' \emph{IEEE Internet of Things Journal}, vol.~9, no.~14,
  pp. 12\,125--12\,146, 2021.

\bibitem{zhang2019networked}
X.-M. Zhang, Q.-L. Han, X.~Ge, D.~Ding, L.~Ding, D.~Yue, and C.~Peng,
  ``Networked control systems: A survey of trends and techniques,''
  \emph{IEEE/CAA Journal of Automatica Sinica}, vol.~7, no.~1, pp. 1--17, 2019.

\bibitem{demirel2016trade}
B.~Demirel, V.~Gupta, D.~E. Quevedo, and M.~Johansson, ``On the trade-off
  between communication and control cost in event-triggered dead-beat
  control,'' \emph{IEEE Transactions on Automatic Control}, vol.~62, no.~6, pp.
  2973--2980, 2016.

\bibitem{peng2018survey}
C.~Peng and F.~Li, ``A survey on recent advances in event-triggered
  communication and control,'' \emph{Information Sciences}, vol. 457, pp.
  113--125, 2018.

\bibitem{yi2018dynamic}
X.~Yi, K.~Liu, D.~V. Dimarogonas, and K.~H. Johansson, ``Dynamic
  event-triggered and self-triggered control for multi-agent systems,''
  \emph{IEEE Transactions on Automatic Control}, vol.~64, no.~8, pp.
  3300--3307, 2018.

\bibitem{heemels2021event}
W.~Heemels, K.~H. Johansson, and P.~Tabuada, ``Event-triggered and
  self-triggered control,'' in \emph{Encyclopedia of Systems and
  Control}.\hskip 1em plus 0.5em minus 0.4em\relax Springer, 2021, pp.
  724--730.

\bibitem{li2020consensus}
X.~Li, Y.~Tang, and H.~R. Karimi, ``Consensus of multi-agent systems via fully
  distributed event-triggered control,'' \emph{Automatica}, vol. 116, p.
  108898, 2020.

\bibitem{xu2023event}
N.~Xu, Z.~Chen, B.~Niu, and X.~Zhao, ``Event-triggered distributed consensus
  tracking for nonlinear multi-agent systems: a minimal approximation
  approach,'' \emph{IEEE Journal on Emerging and Selected Topics in Circuits
  and Systems}, vol.~13, no.~3, pp. 767--779, 2023.

\bibitem{wangsurvey}
Y.~Wang, S.~Wu, C.~Lei, J.~Jiao, and Q.~Zhang, ``A review on wireless networked
  control system: The communication perspective,'' \emph{IEEE Internet of
  Things Journal}, vol.~11, no.~5, pp. 7499--7524, 2024.

\bibitem{masonmulti}
F.~Mason, F.~Chiariotti, A.~Zanella, and P.~Popovski, ``Multi-agent
  reinforcement learning for coordinating communication and control,''
  \emph{IEEE Transactions on Cognitive Communications and Networking}, vol.~10,
  no.~4, pp. 1566--1581, 2024.

\bibitem{liu2024event}
J.~Liu, J.~Ke, J.~Liu, X.~Xie, and E.~Tian, ``Event-driven intelligent dynamic
  positioning for networked unmanned marine vehicles within reinforcement
  learning framework,'' \emph{IEEE Transactions on Vehicular Technology},
  vol.~73, no.~10, pp. 15\,669--15\,674, 2024.

\bibitem{imer2005optimal}
O.~C. Imer and T.~Ba{\c s}ar, ``Optimal estimation with limited measurements,''
  in \emph{Proceedings of the 44th IEEE Conference on Decision and Control},
  2005, pp. 1029--1034.

\bibitem{imer2010optimal}
------, ``Optimal estimation with limited measurements,'' \emph{International
  Journal of Systems, Control and Communications}, vol.~2, no. 1-3, pp. 5--29,
  2010.

\bibitem{lipsa2011remote}
G.~M. Lipsa and N.~C. Martins, ``Remote state estimation with communication
  costs for first-order {LTI} systems,'' \emph{IEEE Transactions on Automatic
  Control}, vol.~56, no.~9, pp. 2013--2025, 2011.

\bibitem{imer2006optimal}
O.~C. Imer and T.~Ba{\c s}ar, ``Optimal control with limited controls,'' in
  \emph{2006 American Control Conference}, 2006, pp. 298--303.

\bibitem{imer2006measure}
O.~C. Imer and T.~Ba\c{s}ar, ``To measure or to control: optimal control with
  scheduled measurements and controls,'' in \emph{2006 American Control
  Conference}, 2006, pp. 1003--1008.

\bibitem{molin2012optimality}
A.~Molin and S.~Hirche, ``On the optimality of certainty equivalence for
  event-triggered control systems,'' \emph{IEEE Transactions on Automatic
  Control}, vol.~58, no.~2, pp. 470--474, 2012.

\bibitem{molin2016optimality}
------, ``Event-triggered state estimation: An iterative algorithm and
  optimality properties,'' \emph{IEEE Transactions on Automatic Control},
  vol.~62, no.~11, pp. 5939--5946, 2017.

\bibitem{maity2020minimal}
D.~Maity and J.~S. Baras, ``Minimal feedback optimal control of
  linear-quadratic-{G}aussian systems: No communication is also a
  communication,'' \emph{IFAC-PapersOnLine}, vol.~53, no.~2, pp. 2201--2207,
  2020.

\bibitem{soleymani2021value}
T.~Soleymani, J.~S. Baras, and S.~Hirche, ``Value of information in feedback
  control: Quantification,'' \emph{IEEE Transactions on Automatic Control},
  vol.~67, no.~7, pp. 3730--3737, 2022.

\bibitem{soleymani2022value}
T.~Soleymani, J.~S. Baras, S.~Hirche, and K.~H. Johansson, ``Value of
  information in feedback control: Global optimality,'' \emph{IEEE Transactions
  on Automatic Control}, 2022.

\bibitem{bacsar1977informationally}
T.~Ba{\c{s}}ar, ``Informationally nonunique equilibrium solutions in
  differential games,'' \emph{SIAM Journal on Control and Optimization},
  vol.~15, no.~4, pp. 636--660, 1977.

\bibitem{bacsar1991game}
------, ``Game theory and {$H^{\infty}$}-optimal control: The continuous-time
  case,'' in \emph{Differential Games—Developments in Modelling and
  Computation: Proceedings of the Fourth International Symposium on
  Differential Games and Applications, Helsinki University of Technology,
  Finland}.\hskip 1em plus 0.5em minus 0.4em\relax Springer, 1991, pp.
  171--186.

\bibitem{bacsar1991optimum}
------, ``Optimum ${H}_\infty$ designs under sampled state measurements,''
  \emph{Systems \& Control Letters}, vol.~16, no.~6, pp. 399--409, 1991.

\bibitem{bacsar1995minimax}
------, ``Minimax control of switching systems under sampling,'' \emph{Systems
  \& Control Letters}, vol.~25, no.~5, pp. 315--325, 1995.

\bibitem{bacsar2005existence}
------, ``On the existence and uniqueness of closed-loop sampled-data nash
  controls in linear-quadratic stochastic differential games,'' in
  \emph{Optimization Techniques: Proceedings of the 9th IFIP Conference on
  Optimization Techniques Warsaw, September 4--8, 1979}.\hskip 1em plus 0.5em
  minus 0.4em\relax Springer, 2005, pp. 193--203.

\bibitem{yuan2017event}
Y.~Yuan, Z.~Wang, and L.~Guo, ``Event-triggered strategy design for
  discrete-time nonlinear quadratic games with disturbance compensations: The
  noncooperative case,'' \emph{IEEE Transactions on Systems, Man, and
  Cybernetics: Systems}, vol.~48, no.~11, pp. 1885--1896, 2017.

\bibitem{maity2016strategies}
D.~Maity and J.~S. Baras, ``Strategies for two-player differential games with
  costly information,'' in \emph{2016 13th International Workshop on Discrete
  Event Systems (WODES)}.\hskip 1em plus 0.5em minus 0.4em\relax IEEE, 2016,
  pp. 211--216.

\bibitem{maity2016optimal}
------, ``Optimal strategies for stochastic linear quadratic differential games
  with costly information,'' in \emph{55th Conference on Decision and
  Control}.\hskip 1em plus 0.5em minus 0.4em\relax IEEE, 2016, pp. 276--282.

\bibitem{maity2017linear}
------, ``Linear quadratic stochastic differential games under asymmetric value
  of information,'' \emph{IFAC-PapersOnLine}, vol.~50, no.~1, pp. 8957--8962,
  2017.

\bibitem{huang2021defending}
Y.~Huang, J.~Chen, and Q.~Zhu, ``Defending an asset with partial information
  and selected observations: A differential game framework,'' in
  \emph{Conference on Decision and Control}, 2021, pp. 2366--2373.

\bibitem{aggarwal2024linear}
S.~Aggarwal, T.~Ba{\c{s}}ar, and D.~Maity, ``Linear quadratic zero-sum
  differential games with intermittent and costly sensing,'' \emph{IEEE Control
  Systems Letters}, vol.~8, pp. 1601--1606, 2024.

\bibitem{aggarwal2024semantic}
S.~Aggarwal, M.~Bastopcu, T.~Ba{\c{s}}ar \emph{et~al.}, ``Semantic
  communication in multi-team dynamic games: A mean field perspective,''
  \emph{arXiv preprint arXiv:2407.06528}, 2024.

\bibitem{bacsar2014stochastic}
T.~Ba{\c{s}}ar, ``Stochastic differential games and intricacy of information
  structures,'' in \emph{Dynamic Games in Economics}.\hskip 1em plus 0.5em
  minus 0.4em\relax Springer, 2014, pp. 23--49.

\bibitem{bacsar1998dynamic}
T.~Ba{\c{s}}ar and G.~J. Olsder, \emph{Dynamic {N}oncooperative {G}ame
  {T}heory}.\hskip 1em plus 0.5em minus 0.4em\relax SIAM, 1998.

\bibitem{bacsar2008h}
T.~Ba{\c{s}}ar and P.~Bernhard, \emph{$H^\infty$ Optimal Control and Related
  Minimax Design Problems: A Dynamic Game Approach}.\hskip 1em plus 0.5em minus
  0.4em\relax Springer Science \& Business Media, 2008.

\bibitem{maity2023efficient}
D.~Maity, ``Efficient communication for pursuit-evasion games with asymmetric
  information,'' in \emph{Proceedings of the 62nd IEEE Conference on Decision
  and Control (CDC)}, 2023, pp. 2104--2109.

\bibitem{baras1989optimal}
J.~S. Baras and A.~Bensoussan, ``Optimal sensor scheduling in nonlinear
  filtering of diffusion processes,'' \emph{SIAM Journal on Control and
  Optimization}, vol.~27, no.~4, pp. 786--813, 1989.

\bibitem{bowers2014introductory}
A.~Bowers and N.~J. Kalton, \emph{An Introductory Course in Functional
  Analysis}.\hskip 1em plus 0.5em minus 0.4em\relax Springer, 2014.

\bibitem{jarlebring2018krylov}
E.~Jarlebring, G.~Mele, D.~Palitta, and E.~Ringh, ``Krylov methods for low-rank
  commuting generalized {S}ylvester equations,'' \emph{Numerical Linear Algebra
  with Applications}, vol.~25, no.~6, p. e2176, 2018.

\end{thebibliography}
